\documentclass[pdflatex,sn-mathphys-num]{sn-jnl}% Math and Physical Sciences Numbered Reference Style 
\usepackage{graphicx}%
\usepackage{amsmath,amssymb,amsfonts,amsthm}%

\usepackage{tikz}

\usetikzlibrary{arrows.meta}
\usetikzlibrary{backgrounds}
\pgfdeclarelayer{background}
\pgfsetlayers{background, main}

\tikzset{main node/.style={rectangle,fill=blue!20,draw=none,minimum size=0.3cm,inner sep=3pt, rounded corners = 3pt,font=\sffamily},
reduced node/.style={main node, fill=brown!20},
rededge/.style={brown}}

\tikzset{transparent node/.style={circle,fill=white!20,draw,minimum size=0cm,inner sep=0pt},}

\tikzset{%
	>={Latex[width=2mm,length=2mm]},
	base/.style = {rectangle, rounded corners, draw=black,
		minimum width=4cm, minimum height=1cm,
		text centered, font=\sffamily, text width=4cm },
	ontology/.style={rectangle, rounded corners, draw=black,
		minimum width=1cm, minimum height=.5cm,
		text centered, font=\sffamily, text width=1cm },
	activityStarts/.style = {base, fill=blue!30},
	entity/.style = {base},
	workflow/.style={rectangle, draw=black,
		minimum width=2cm, minimum height=0.4cm,
		text centered, font=\sffamily, text width=4cm},
	data/.style={rectangle, rounded corners, draw=black,
		minimum width=2cm, minimum height=0.4cm,
		text centered, font=\sffamily, text width=5cm},
	startstop/.style = {base, fill=red!30},
	activityRuns/.style = {base, fill=green!30},
	process/.style = {base, minimum width=2.3cm, fill=orange!15,
		font=\ttfamily},
}

\usetikzlibrary{arrows.meta}
\usetikzlibrary{backgrounds}

\usetikzlibrary{calc, intersections}
\usetikzlibrary{arrows, shapes}
\usetikzlibrary{positioning}

\usetikzlibrary{matrix,decorations.pathreplacing, positioning,fit}
\usepackage{tkz-graph}

\usetikzlibrary{shapes.geometric,calc}

\pgfsetlayers{background, main}

\usepackage{mathrsfs}%
\usepackage[title]{appendix}%
\usepackage{xcolor}%
\usepackage{textcomp}%
\usepackage{manyfoot}%
\usepackage{booktabs}%
\usepackage[noend]{algpseudocode}%

\usepackage{thm-restate}
\usepackage{hyperref}       % hyperlinks
\usepackage{url}            % simple URL typesetting
\usepackage{booktabs}       % professional-quality tables
\usepackage{mathtools}
\usepackage{adjustbox}
\usepackage{microtype}    
\usepackage{fancyhdr}
\usepackage{xspace,mfirstuc,tabulary}
\usepackage{amssymb}
\usepackage{bm}
\usepackage{subfig}
\usepackage{bbm}
\usepackage{algorithm}
\usepackage[title]{appendix}
\usepackage{array,multirow,makecell} 
\usepackage{multicol}
\usepackage{todonotes}
\usepackage{cleveref}
\usepackage{tikz}
\usepackage{array}
\usepackage{lmodern} %remove warning about missing font size

\usetikzlibrary{calc, intersections}
\usetikzlibrary{arrows, shapes, fit}
\usetikzlibrary{positioning}
\usetikzlibrary{arrows.meta}
\usetikzlibrary{matrix,decorations.pathreplacing, calc, positioning,fit}
\usepackage{tkz-graph}
\usetikzlibrary{arrows, shapes, fit}
\usetikzlibrary{shapes.geometric,calc}
\usetikzlibrary{arrows.meta}
\tikzset{main node/.style={rectangle,fill=blue!20,draw=none,minimum size=0.3cm,inner sep=3pt, rounded corners = 3pt,font=\sffamily},
reduced node/.style={main node, fill=brown!20},
rededge/.style={brown}}

\tikzset{transparent node/.style={circle,fill=white!20,draw,minimum size=0cm,inner sep=0pt},}

\tikzset{%
	>={Latex[width=2mm,length=2mm]},
	base/.style = {rectangle, rounded corners, draw=black,
		minimum width=4cm, minimum height=1cm,
		text centered, font=\sffamily, text width=4cm },
	ontology/.style={rectangle, rounded corners, draw=black,
		minimum width=1cm, minimum height=.5cm,
		text centered, font=\sffamily, text width=1cm },
	activityStarts/.style = {base, fill=blue!30},
	entity/.style = {base},
	workflow/.style={rectangle, draw=black,
		minimum width=2cm, minimum height=0.4cm,
		text centered, font=\sffamily, text width=4cm},
	data/.style={rectangle, rounded corners, draw=black,
		minimum width=2cm, minimum height=0.4cm,
		text centered, font=\sffamily, text width=5cm},
	startstop/.style = {base, fill=red!30},
	activityRuns/.style = {base, fill=green!30},
	process/.style = {base, minimum width=2.3cm, fill=orange!15,
		font=\ttfamily},
}

\usetikzlibrary{arrows.meta}
\usetikzlibrary{backgrounds}
\pgfdeclarelayer{background}
\pgfsetlayers{background, main}

\newcolumntype{L}[1]{>{\raggedright\let\newline\\\arraybackslash\hspace{0pt}}m{#1}}
\newcolumntype{C}[1]{>{\centering\let\newline\\\arraybackslash\hspace{0pt}}m{#1}}
\newcolumntype{R}[1]{>{\raggedleft\let\newline\\\arraybackslash\hspace{0pt}}m{#1}}

\newcommand{\statement}[1]{%
	#1\enspace\ignorespaces
}

\def \n {8}

\def \verySmallRadius {0.95cm}
\def \margin {10} % margin in angles, depends on the radius

\newcounter{mylistcounter}

\def\saveitem#1{%
	\stepcounter{mylistcounter}%
	\expandafter\def\csname mylist\themylistcounter\endcsname{#1}}

\forcsvlist{\saveitem}{%
	31,24,23,43,42,41,34,32,C
}%

\def\getnthelement#1{\csname mylist#1\endcsname}
\newcounter{mylistcounterbis}
\def\saveitembis#1{%
	\stepcounter{mylistcounterbis}%
	\expandafter\def\csname mylistbis\themylistcounterbis\endcsname{#1}}

\forcsvlist{\saveitembis}{%
	341, 342, 423, 234, 432, 324, 431, C
}%

\def\getnthelementbis#1{\csname mylistbis#1\endcsname}

\newcounter{mylistcounterthird}
\def\saveitemthird#1{%
	\stepcounter{mylistcounterthird}%
	\expandafter\def\csname mylistthird\themylistcounterthird\endcsname{#1}}

\forcsvlist{\saveitemthird}{%
	14,24,41,12,23,33,13,21,42, C
}%

\def\getnthelementthird#1{\csname mylistthird#1\endcsname}
\newcounter{mylistcounterfourth}
\def\saveitemfourth#1{%
	\stepcounter{mylistcounterfourth}%
	\expandafter\def\csname mylistfourth\themylistcounterfourth\endcsname{#1}}

\forcsvlist{\saveitemfourth}{%
	213,133,233,123,412,241,124, 214, 421, 142, C
}%

\def\getnthelementfourth#1{\csname mylistfourth#1\endcsname}

\newcommand{\Def}[1]{{\bfseries #1}}

\newcolumntype{P}[1]{>{\centering\arraybackslash}p{#1}}

\newtheorem{definition}{Definition}
\newtheorem{problem}{Problem}
\newtheorem{corollary}{Corollary}
\newtheorem{proposition}{Proposition}
\newtheorem{lemma}{Lemma}
\newtheorem{theorem}{Theorem}
\newtheorem{remark}{Remark}
\newtheorem{claim}{Claim}

\newcommand{\straightsf}[1]{\normalfont {\sf #1}\xspace}
\newcommand{\GU}{\texorpdfstring{{\straightsf{GU}}}{GU}\xspace}
\newcommand{\GW}{\texorpdfstring{{\straightsf{GW}}}{GW}\xspace}
\newcommand{\DU}{\texorpdfstring{{\straightsf{DU}}}{DU}\xspace}
\newcommand{\DW}{\texorpdfstring{{\straightsf{DW}}}{DW}\xspace}

\hypersetup{
  colorlinks   = true, %Colours links instead of ugly boxes
  urlcolor     = blue, %Colour for external hyperlinks
  linkcolor    = blue, %Colour of internal links
  citecolor   = red %Colour of citations
}

\title[Sequence graphs realizations and ambiguity in language models]{Sequence Graphs Realizations and Ambiguity in Language Models\footnote{A preliminary version of this work has been published in COCOON'21.}}

\newcommand\blfootnote[1]{%
  \begingroup
  \renewcommand\thefootnote{}\footnote{#1}%
  \addtocounter{footnote}{-1}%
  \endgroup
}

\author[1,2]{\fnm{Sammy} \sur{Khalife} \blfootnote{Corresponding author: Sammy Khalife, \href{mailto:khalife.sammy@cornell.edu}{khalife.sammy@cornell.edu}}}
\author[3]{\fnm{Yann} \sur{Ponty}}
\author[4]{\fnm{Laurent} \sur{Bulteau}}

\affil[1]{\orgdiv{Department of Applied Mathematics and Statistics}, \orgname{Johns Hopkins University}, \country{USA}}%, \orgaddress{\street{Street}, \city{City}, \postcode{100190}, \state{State}, \country{Country}}}

\affil[2]{\orgdiv{School of Operations Research and Information Engineering, Cornell Tech}, \orgname{Cornell University}, \country{USA}}%, 

\affil[3]{\orgdiv{LIX}, \orgname{CNRS UMR 7161, Ecole Polytechnique, Institut Polytechnique de Paris}, \country{France}} %, \orgaddress{\street{Street}, \city{City}, \postcode{10587}, \state{State}, \country{Country}}

\affil[4]{\orgdiv{LIGM}, \orgname{CNRS, Universite Gustave Eiffel}, \country{France}} %\orgaddress{\street{} \city{  Marne-la-Vallee}, \postcode{77454}, \state{State}, \country{Country}}
\abstract{Several popular language models represent local contexts in an input text $x$ as bags of words. Such representations are naturally encoded by a sequence graph whose vertices are the distinct words occurring in $x$, with edges representing the (ordered) co-occurrence of two words within a sliding window of size $w$. However, this compressed representation is not generally bijective: some may be ambiguous, admitting several realizations as a sequence, while others may not admit any realization.

In this paper, we study the realizability and ambiguity of sequence graphs from a combinatorial and algorithmic point of view. We consider the existence and enumeration of realizations of a sequence graph under multiple settings: window size $w$, presence/absence of graph orientation, and presence/absence of weights (multiplicities). When $w=2$, we provide polynomial time algorithms for realizability and enumeration in all cases except the undirected/weighted setting, where we show the \#P-hardness of enumeration. For $w\ge 3$, we prove the hardness of all variants, even when  $w$ is considered as a constant, with the notable exception of the undirected unweighted case for which we propose XP algorithms for both problems. % (realizability and enumeration) problems. %, tight due to a corresponding W[1]-hardness result. 
We conclude with an integer program formulation to solve the realizability problem, and a dynamic programming algorithm to solve the enumeration problem in instances of moderate sizes. This work leaves open the membership to NP of both problems, a non-trivial question due to the existence of minimum realizations having size exponential on the instance encoding.}


\begin{document}

\maketitle

%\section*{Contents}

%\tableofcontents

%\newpage

\section{Introduction}
%\blfootnote{A preliminary version of this work has been published in COCOON 21}

%The automated treatment of familiar objects, either natural or artificial, always relies on a translation into entities manageable by computer programs. 

%\todo[inline]{NEED FOR A NEW INTRO}

%\input{intro_beginning}

The construction of  numerical vector representations for words and sentences (a.k.a \textit{embeddings}) %, for a sequence of words from a vocabulary is
has been a long-standing problem in data science and artificial intelligence. Until the mid 2010s, most of the techniques in information retrieval and Natural Language Processing (NLP)  used pre-computed embedding as input to machine learning algorithms. In this context, the choice and computation of the ``right'' embedding constitutes a problem in its own right. Ultimately, each embedding of this type can be thought of as a map between words to numerical vectors: each word is attributed a \emph{static} vector, hopefully encoding different types of information about the word (semantic, spelling, ...). In a new document, this word is mapped to the same vector, independently of the new context in which it appears. We refer the reader to \cite{grohe2020word2vec} for a description of these embedding techniques, such as \emph{Word2Vec}, \emph{GloVe}, \emph{FastText}, and others. These embeddings laid the groundwork for modern NLP techniques, by being scalable, interpretable, and performing for several information retrieval tasks. On the other hand, attention-based models \cite{vaswani2017attention}, and its various implementations such as \emph{transformers} \cite{phuong2022formal}, adopt a different paradigm by incorporating both the embeddings and the learning parameters in the same computational framework. Furthermore, they allow the embedding of each word to \emph{depend} on the context it appears in. This new paradigm combined with training techniques lead to state-of-the-art methods in most NLP tasks including complex ones that traditional embeddings were not able to handle well, in particular translation, summarization, question answering, and natural text understanding. 
%They still provide an alternative to methods such as transformers by as they  complement for integration, scalability, efficiency and interpretability. Researching the combinatorics of these embeddings helps deepen our understanding of their underlying principles and the ways in which words are represented in vector spaces \cite{grohe2020word2vec}. 

Measuring the capacity of  embeddings to faithfully represent words or sentences is a natural problem in order to understand their limitations. In particular, one may wonder and ask about their level of \emph{ambiguity}: \emph{How many changes can be made to a sentence that would yield the same embedding? How long does a sequence need to be such that the resulting embedding is not unique?  } More generally, the question of invariance (and equivariance\footnote{Informally, embeddings are equivariant if under some group action,  changes in the input sequence should affect the embeddings in the same manner.}) in computational linguistics  -- including embeddings --  has become a topic of growing interest: in a recent line of research initiated by \cite{gordon2019permutation,peyrard2021invariant}, the authors introduce the idea that  \emph{equivariance} in embeddings and  can also be a desirable property. Namely, Gordon et al. \cite{gordon2019permutation} exposed a method to construct \emph{string equivariance}, which suppose some form of equivariance under "local" permutation action. White et al. \cite{white2022equivariant} further generalize this work by restricting the range of permutations  only  to some lexical classes: symmetries are only allowed to exchange words inside a given fixed lexical class (examples of lexical classes are subsets of verbs, adjectives, nouns of interchangeable semantic roles). This line of work,   builds on the premise that certain embeddings exhibit desirable properties, enabling them to emulate fundamental cognitive abilities such as \emph{compositional generalization}\footnote{In a nutshell, compositional generalization is the  ability to process a novel sentence and assign overall possible meanings (e.g. ``cook twice'') by composing the meanings of its individual parts (e.g. meanings of ``cook'' and ``twice'').} \cite{chomsky2002syntactic,montague1970universal,partee1995lexical}. In turn, embeddings sharing such properties would lead to general computational frameworks coming closer to a true understanding of words and sentences \cite{petrache2024position}.%  Interestingly  also argues that group equivariance is not sufficient to emulate compositional generalization, and that the proper type of equivariance is . 

In this article we study a family of combinatorial problems related to the  first group of \emph{context-based} embeddings (including \emph{Word2Vec}, \emph{Glove}, \emph{FastText}, ...). Some of these problems can be naturally interpreted as determining the level of \emph{ambiguity} of these embeddings. While context-based embeddings are richer than the traditional \emph{bag-of-words} (related to  \textit{Parikh vectors}) in the literature, they still induce some level of ambiguity, \emph{i.e.} a given graph can represent several sequences (see Figure \ref{counter_example_dw_2a} and \ref{counter_example_dw_2b} for illustrations). The main contribution of this article is to present new complexity results about the inverse problems quantifying such level of ambiguity. We also  present algorithms that allow to solve these problems of graphs of reasonable size, despite the computational hardness in most variations of the problems. Our results also highlight the potential limitations of  context-based embedding to encode the  whole information of sequences of reasonable length.  In the line of recent efforts on invariance and equivariance in computational linguistics, we hope that this study can be connected to intrinsic desirable properties of word embeddings, and serve as a stepping stone to further understand their attention-based variants \cite{peyrard2021invariant}.
%An interesting avenue of research would be to connect the complexity of evaluating the ambiguity or invariance of an embedding, to its ability to emulate some fundamental capacity, such as compositional generalization.

% The remaining of the article is organized as follows. In Section~\ref{sec:defs} we present the definitions and statement of the problems studied in this article. Section \ref{subsec:relatedwork} present existing work on related inverse problems. In Section~\ref{sec:pres_results} we present our main theoretical results. Full proofs are given in Sections \ref{sec:proofsTwoSeq} ($w=2$) and \ref{w_3_study} ($w\geq 3$). In Section \ref{sec:effective_general_algo}, we propose an integer program and a dynamic programming algorithm to respectively recognize a sequence graph and count its realizations. Finally, Section \ref{sec:exponential_realizations} concludes with some comments and discussions. % and a first step towards an answer to the belonging of our problems to NP, by proving the existence of graphs whose minimal realizations have exponential size.

\subsection{Definitions and problem statement}\label{sec:defs}
In the following, $p $ is a positive integer and  $[p] $ is a shorthand for $ \{1,...,p\}$. Let $x = x_1, x_2, ..., x_p$ be a finite sequence over a vocabulary $X = \{v_1, \cdots, v_n\}$. %Without loss of generality, we suppose that $X = $. 
We first formalize the notion of \emph{sequence graph} --  introduced in \cite{rousseau2015text}  as \emph{graph-of-words} -- illustrated in Figures~\ref{counter_example_dw_2a} and \ref{counter_example_dw_2b}. %Those are formal definitions of the constructions i following the . %\todo{reference to some paper introducing this def ?}

%\begin{definition}\label{definition_sequence_graphs}
% $G=(V,E)$ is the sequence graph (or $w$-sequence graph) of the sequence $x$ with window size $w \in \mathbb{N}^{+}$ ($w > 0$) if and only if $V=\{v\in X \, | \,  \exists i\in[p], \; v=x_i\}$, and
% \begin{equation}\label{undirected_seq_graph_def}
%    (u,v)\in E \iff \exists (k,k')\in [p]^2\quad  0<|k-k'|\leq w-1, \; u=x_k, \; v=x_{k'}
% \end{equation}
% 
% A sequence graph $G$ is endowed with a weight matrix $\Pi(G)=(\pi_{ij})$ such that
% \begin{equation}\label{undirected_weighted_seq_graph_def}
%   \pi_{ij} = \mathsf{Card}\;\{(k,k') \in [p] ^2 \;|\;  \; 0< |k-k'| \leq w - 1,  \;  x_k =i \; \mathsf{and} \; x_{k'}=j\}
% \end{equation}
% 
% For digraphs, the absolute values in Statements \eqref{undirected_seq_graph_def} and \eqref{undirected_weighted_seq_graph_def} are replaced with $k < k' \leq k + w - 1$.
% 
%We say that $x$ is a $w$-realization of $G$, or realization in the absence of ambiguity, if $G$ is the $w$-sequence graph of $x$. Finally, $x$ is a $w$-realization of $(G, \Pi)$ if $G$ is the $w$-sequence graph of $x$ with $\Pi(G)=\Pi$.
%\end{definition}

\begin{definition}\label{definition_sequence_graphs}
    Given a sequence $x$ and a window size $w \in \mathbb{N}^{+}$ ($w > 0 $), the \emph{sequence graph of  $x$ with window size $w$} is the graph where each vertex of $G=(V,E)$ is the set of distinct words appearing in the sequence, and each edge notifies the appearance of two words in a context of size $w$. Formally, 
$V=\{v\in X \; | \;  \exists i\in[p], \; v=x_i\}$
and
\begin{equation}\label{undirected_seq_graph_def}\{u,v\}\in E \iff \exists (k,k')\in [p]^2\quad  0<|k-k'|\leq w-1, \; u=x_k, \; v=x_{k'}\end{equation}
  A weighted sequence graph $G$ is endowed with a weight matrix $\Pi(G)=(\pi_{e})_{e\in E}$ such that
  \begin{equation}\label{undirected_weighted_seq_graph_def}
    \pi_{\{u,v\}} = \mathsf{Card}\;\{(k,k') \in [p] ^2 \;|\;  \; 0< |k-k'| < w ,  \;  x_k =u \; \mathrm{and} \; x_{k'}=v\}
  \end{equation}

For digraphs, the two-element set $\{u,v\}$ is replaced by the pair $(u,v)$ in Statement \eqref{undirected_seq_graph_def}, and the absolute values in Statements \eqref{undirected_seq_graph_def} and \eqref{undirected_weighted_seq_graph_def} are replaced with $k < k' < k + w $.

   Under these conditions, we say that $x$ is a \emph{$w$-realization} of $G$, or \emph{realization }in the absence of ambiguity, if $G$ is the sequence graph of $x$ with window size $w$. Finally, $x$ is a \emph{$w$-realization} of $(G, \Pi)$ if $G$ is the $w$-sequence graph of $x$ with $\Pi(G)=\Pi$. 
\end{definition}

\begin{figure}[H]
    \centering
    \newcommand{\HSep}{.52cm}
    \subfloat[Unambiguous graph ($w=3$)]{
    \begin{tikzpicture}
    \begin{scope}[]%[xshift=6cm]
    \node[main node] (1) {Linux};
    \node[main node] (2) [right= \HSep  of 1] {is};
    \node[main node] (3) [right= \HSep  of 2] {not};
    \node[main node] (4) [right= \HSep  of 3 ] {UNIX};
    \node[main node] (5)  [right= \HSep  of 4]{but};
    
    \draw[->] (1) edge node[right] {} (2);
    \draw[->] (1) edge [bend left=30] node{} (3);
    
    \draw[->] (3) edge [bend left=30] node{} (5);
    
    \draw[->] (2) edge node[right] {} (3);
    \draw[->] (2) edge [bend left=30] node{} (4);
    \draw[->] (3) edge node[right] {} (4);
    
    \draw[->] (4) edge node[right] {} (5);
    
    \draw[->] (4) edge [bend left=22] node {} (1);
    \draw[->] (5) edge [bend left=30] node{} (1);   
    
    \end{scope}
    \end{tikzpicture}

    }\hfill
    \subfloat[Ambiguous graph ($w=2$)]{
    \begin{tikzpicture}
    \begin{scope}[]%[xshift=6cm]
    \node[main node] (1) {Linux};
    \node[main node] (2) [right= \HSep  of 1] {is};
    \node[main node] (3) [right= \HSep  of 2] {not};
    \node[main node] (4) [right= \HSep  of 3 ] {UNIX};
    \node[main node] (5)  [right= \HSep  of 4]{but};

    \draw[->] (1) edge node[right] {} (2);
    \draw[->] (2) edge node[right] {} (3);
    \draw[->] (3) edge node[right] {} (4);
    
    \draw[->] (4) edge node[right] {} (5);
    \draw[->] (5) edge [bend left=30] node{} (1);   
  
    \end{scope}
    \end{tikzpicture}

    }
    \caption{Sequence digraphs (or directed  \textit{graphs-of-words})  built for the sentence ``{Linux is not UNIX but Linux}'' using window sizes $w=3$ (a) and $w=2$ respectively (b). In the second case, the sequence graph is ambiguous, since any circular permutation of the words admits the same representation.}
    \label{counter_example_dw_2a}
  \end{figure}
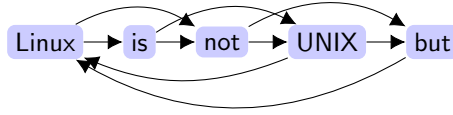
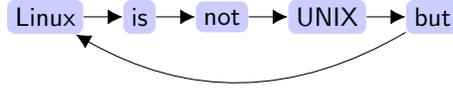
  
    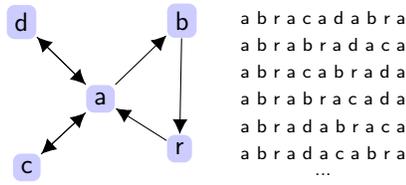
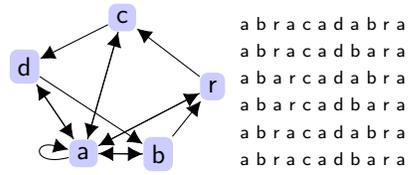
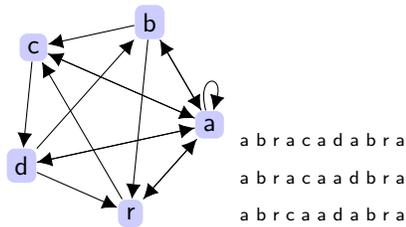
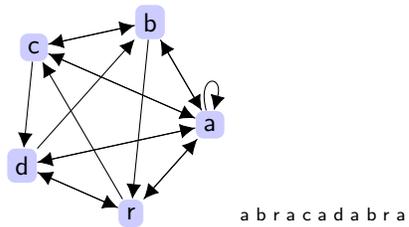
\begin{figure}[H]
    \centering
    \newcommand{\HSep}{.52cm}
    \subfloat[$w=2$, $G$ has  $30$ realizations]{
\begin{tikzpicture}
\begin{scope}[xshift=4cm]
\node[main node] at (-0.004843,-0.161971) (0) {a};
\node[main node] at (1.07740,0.878132) (1) {b};
\node[main node] at (-0.984928,-1.074439) (2) {c};
\node[main node] at (-1.054361,0.86) (3) {d};
\node[main node] at (1.045639,-0.824727) (4) {r};
\draw[->] (0) edge node[right] {} (1);
\draw[->] (0) edge node[right] {} (2);
\draw[->] (0) edge node[right] {} (3);
\draw[->] (1) edge node[right] {} (4);
\draw[->] (2) edge node[right] {} (0);
\draw[->] (3) edge node[right] {} (0);
\draw[->] (4) edge node[right] {} (0);
\end{scope}
\end{tikzpicture}
\quad
\begin{tikzpicture}
\newcommand{\VSep}{0pt}
\sf 
\scriptsize
%[, 'a b r a b r a d a c a', 'a b r a c a b r a d a', 'a b r a c a d a b r a', 'a b r a d a b r a c a', 'a b r a d a c a b r a']
\node at (0,3) (Sentence) {a b r a c a d a b r a};
    \node[below=\VSep of Sentence] (Sentence) { a b r a b r a d a c a};
     \node[below=\VSep of Sentence]  (Sentence) {a b r a c a b r a d a};
      \node[below=\VSep of Sentence]  (Sentence) {a b r a b r a c a d a};
       \node[below=\VSep of Sentence]  (Sentence) {a b r a d a b r a c a};
        \node[below=\VSep of Sentence]  (Sentence) {a b r a d a c a b r a};
        \node[below=\VSep of Sentence] (Sentence) {...};
\end{tikzpicture}
    }\hfill
    \subfloat[$w=3$, $G$ has $6$ realizations]{
\begin{tikzpicture}
\begin{scope}[xshift=4cm]
\node[main node] at (-0.51647,-0.992638) (0) {a};
\node[main node] at (0.47779,-1) (1) {b};
\node[main node] at (0,0.817444) (2) {c};
\node[main node] at (-1.300001,0.163526) (3) {d};
\node[main node] at (1.045639,0.163526) (4) {r};
\draw[->, loop left] (0) to (0);
\draw[->] (0) edge node[right] {} (1);
\draw[->] (0) edge node[right] {} (2);
\draw[->] (0) edge node[right] {} (3);
\draw[->] (0) edge node[right] {} (4);
\draw[->] (1) edge node[right] {} (0);
\draw[->] (1) edge node[right] {} (4);
\draw[->] (2) edge node[right] {} (0);
\draw[->] (2) edge node[right] {} (3);
\draw[->] (3) edge node[right] {} (0);
\draw[->] (3) edge node[right] {} (1);
\draw[->] (4) edge node[right] {} (0);
\draw[->] (4) edge node[right] {} (2);
\end{scope}
\end{tikzpicture}
\begin{tikzpicture}\sf 
\scriptsize
%[, 'a b r a b r a d a c a', 'a b r a c a b r a d a', 'a b r a c a d a b r a', 'a b r a d a b r a c a', 'a b r a d a c a b r a']
\newcommand{\VSep}{0pt}
\node at (0,3) (Sentence) {a b r a c a d a b r a};
    \node[below=\VSep of Sentence] (Sentence) {a b r a c a d b a r a};
     \node[below=\VSep of Sentence]  (Sentence) {a b a r c a d a b r a};
      \node[below=\VSep of Sentence]  (Sentence) {a b a r c a d b a r a};
       \node[below=\VSep of Sentence]  (Sentence) {a b r a c a d a b r a};
        \node[below=\VSep of Sentence]  (Sentence) {a b r a c a d b a r a};
\end{tikzpicture}
    }
    
    \subfloat[$w=4$, $G$ has $3$ realizations]{
\begin{tikzpicture}
\begin{scope}[xshift=4cm]
\node[main node] at (1.086105,-0.415112) (0) {a};
\node[main node] at (0.288890,0.950456) (1) {b};
\node[main node] at (-1.256194,0.614243) (2) {c};
\node[main node] at (-1.413895,-0.959116) (3) {d};
\node[main node] at (0.033726,-1.595292) (4) {r};
\draw[->, loop above] (0) to (0);
\draw[->] (0) edge node[right] {} (1);
\draw[->] (0) edge node[right] {} (2);
\draw[->] (0) edge node[right] {} (3);
\draw[->] (0) edge node[right] {} (4);
\draw[->] (1) edge node[right] {} (0);
\draw[->] (1) edge node[right] {} (2);
\draw[->] (1) edge node[right] {} (4);
\draw[->] (2) edge node[right] {} (0);
\draw[->] (2) edge node[right] {} (3);
\draw[->] (3) edge node[right] {} (0);
\draw[->] (3) edge node[right] {} (1);
\draw[->] (3) edge node[right] {} (4);
\draw[->] (4) edge node[right] {} (0);
\draw[->] (4) edge node[right] {} (2);
\end{scope}
\end{tikzpicture}
\begin{tikzpicture}\sf
\scriptsize
\node at (0,3) (Sentence) {a b r a c a d a b r a};
    \node at (0,2.5) (Sentence) {a b r a c a a d b r a};
    \node at (0, 2) (Sentence) {a b r c a a d a b r a};
\end{tikzpicture}
    }
    \hfill
    \subfloat[$w=5$, $G$ has one realization]{
\begin{tikzpicture}
\begin{scope}[xshift=4cm]
\node[main node] at (1.086105,-0.415112) (0) {a};
\node[main node] at (0.288890,0.950456) (1) {b};
\node[main node] at (-1.256194,0.614243) (2) {c};
\node[main node] at (-1.413895,-0.959116) (3) {d};
\node[main node] at (0.033726,-1.595292) (4) {r};
\draw[->, loop above] (0) to (0);
\draw[->] (0) edge node[right] {} (1);
\draw[->] (0) edge node[right] {} (2);
\draw[->] (0) edge node[right] {} (3);
\draw[->] (0) edge node[right] {} (4);
\draw[->] (1) edge node[right] {} (0);
\draw[->] (1) edge node[right] {} (2);
\draw[->] (1) edge node[right] {} (4);
\draw[->] (2) edge node[right] {} (0);
\draw[->] (2) edge node[right] {} (1);
\draw[->] (2) edge node[right] {} (3);
\draw[->] (3) edge node[right] {} (0);
\draw[->] (3) edge node[right] {} (1);
\draw[->] (3) edge node[right] {} (4);
\draw[->] (4) edge node[right] {} (0);
\draw[->] (4) edge node[right] {} (2);
\draw[->] (4) edge node[right] {} (3);
\end{scope}
\end{tikzpicture}
\begin{tikzpicture} \sf 
\scriptsize
\node at (0,3) (Sentence) {a b r a c a d a b r a};
\end{tikzpicture}
    }   
    \caption{Sequence digraphs (or directed \textit{graphs-of-words})  built for the sentence ``{a b r a c a d a b r a}'' using window sizes 2 (a), 3 (b), 4 (c) and 5 (d). }
    \label{counter_example_dw_2b}
    %}
  \end{figure}
 
 Given $w$, the graph of a sequence $x$ is unique and the natural integers $\pi_{\{u,v\}}$ represent the number of co-occurrences of $u$ and $v$ in all windows of size $w$.  %A linear time algorithm 
 An algorithm to construct a weighted sequence graph is presented in Algorithm \ref{algo:graph_sequence}; the other cases (unweighted, directed) are obtained similarly. In the unweighted case, the map thus defined from the sequence set  $X^{\star}$ to the graph set $\mathcal{G}$ is referred to as $\phi_w \colon X^\star \to \mathcal{G}, x \mapsto G_w(x)$. Based on these definitions, we consider the following problems:

  \begin{algorithm}[t]
  \caption{ Construction of $\Pi$ associated to a weighted sequence graph}
  \label{algo:graph_sequence}
  \textbf{Parameter}: Window size $w \geq 2$\\
  \textbf{Input}: Sequence $x$ of length $p \geq 1$ of integers in $[1,n]$ \\
  %\textbf{Parameter}: Optional list of parameters\\
  \textbf{Output}: Weighted adjacency matrix $\Pi=(\pi_{\{u,v\}})$ %$(G_w(x), \, \Pi)$
  \begin{algorithmic}[1] 
  \State $\pi_{\{u,v\}}\gets 0$ for $u,v\in [n]^2$
  
  \For{$ i = 1$  to $p-1$}
        \For{$ j = i+1$ to $\min(i+w-1, p)$}
     \State $\pi_{\{x_i,x_{j}\}} \leftarrow \pi_{\{x_i,x_j\}}+1$
    \EndFor
  \EndFor
  \State \textbf{return} $\Pi=(\pi_{\{u,v\}})$
  \end{algorithmic}
 \end{algorithm}

\newcommand{\Clique}{{\text{\sf Clique}}}
\newcommand{\RealizableG}{\ensuremath{\text{\sf Realizable}}}
\newcommand{\NumRealizationsG}{\ensuremath{\text{\sf NumRealizations}}}

\newcommand{\Realizable}[1]{\texorpdfstring{\ensuremath{\text{\sf Realizable}_{#1}}}{Realizable #1}}
\newcommand{\NumRealizations}[1]{\texorpdfstring{\ensuremath{\text{\sf NumRealizations}_{#1}}}{NumRealizations #1}}
%\begin{problem}[\Clique \label{Decision_general}]
\newcommand{\Problem}[4]{
\begin{problem}[#1\label{#2}]
~~\\  
  \noindent{\bfseries Input:} #3
  
  \noindent{\bfseries Output:} #4
\end{problem}
}

\Problem{Weighted-\RealizableG{}  (W-\RealizableG)}%
{Decision_general_W}%
{Graph $G$ (directed or undirected), weight matrix $\Pi$, window size $w$}%
{True if $(G,\Pi)$ admits a $w$-realization $x$, False otherwise.}

\Problem{Unweighted-\RealizableG\  (U-\RealizableG)}%
{Decision_general_U}%
{Graph $G$ (directed or undirected), window size $w$}%
{True if $G$ admits a $w$-realization $x$, False otherwise.}
We denote by \DW-\RealizableG\ (resp. {\GW-}) the restricted version of W-\RealizableG\  where the input graph $G$ is directed (resp. undirected). We define \GU-\Realizable\ and \DU-\Realizable\ similarly. Our  notations can be summarized as:

\smallskip\noindent
    \begin{tabular}{l@{\qquad}l}
  \DW $\to$ (directed) Digraph,  Weighted &  \DU $\to$ (directed) Digraph,  Unweighted \\
  \GW $\to$ (undirected) Graph, Weighted  &  \GU $\to$    (undirected) Graph, Unweighted
    \end{tabular}

\Problem{Unweighted-\NumRealizationsG{} (U-\NumRealizationsG)}%
{Counting_general_U}%
{Graph $G$ (directed or undirected), window size $w$}%
{The number of \Def{realizations} of $G$, \emph{i.e.} preimages of $G$ through $\phi_w$ i.e. $|\{x\in X^\star \mid \phi_w(x) = G \}|$
  if finite, or $+\infty$ otherwise.}%

\Problem{Weighted-\NumRealizationsG{} (W-\NumRealizationsG)}%
{Counting_general_W}%
{Graph $G$ (directed or undirected), weight matrix $\Pi$, window size $w$}%
{The number of \Def{realizations} of $G$ in the weighted sense.}

%\begin{problem}[\Realizable{w}\label{Decision}]
% ~\\
% \noindent{\bfseries Parameters:} Window size $w$
% 
% \noindent{\bfseries Input:} Graph $G$ (and optional matrix weights $\Pi$)
% 
% \noindent{\bfseries Output:} True if $(G, \Pi)$ is the $w$-sequence graph of some sequence $x$, False otherwise.
% 
%\end{problem}

%\begin{problem}[\NumRealizations{w}\label{Counting}]
% ~\\
% \noindent{\bfseries Parameters:} Window size $w$
% 
% \noindent{\bfseries Input:} Graph $G$ (and optional matrix weights $\Pi$)%, window size $w$
% 
% \noindent{\bfseries Output:} The number of \Def{realizations} of $G$, \emph{i.e.} preimages of $G$ through $\phi_w$ such that $|\{x\in X^\star \mid \phi_w(x) = G \}|$
% if finite, or $+\infty$ otherwise.
%\end{problem}
Similarly, we use the same prefix for the directed or undirected versions (\DW,\DU, \GW, \GU). We also denote \NumRealizations{w} for the case where $w$ is a fixed positive integer. Note that \NumRealizationsG{} generalizes the previous one, as \RealizableG{} can be solved by testing the nullity of the output of \NumRealizations{}.

\subsection{Related work}\label{subsec:relatedwork}
Sequence graphs encode the information in several models based on co-occurences~\cite{rousseau2015text,arora2016latent,pennington2014glove}. 
To the best of our knowledge, the ambiguity and realizability questions addressed in this work were never addressed by prior work in computational linguistics. %Furthermore, we believe the problems studied in this paper are new and interesting from an algorithmic point of view, and appear to be devoid of  reduction to other well-known problems. 
%, and also to compare how this invariance is shared with neural networks. 
It may seem that the inverse problems we are considering are similar to the \textit{Universal Reconstruction of a String} \cite{gawrychowski2020universal}, which consists in determining the set of strings of a fixed length having as many distinct letters as possible, satisfying substring equations of the form: 
$s[q_1 \cdots q_{p}] = s[q'_1 \cdots q'_{p}], \, \cdots \,  , s[r_1 \cdots r_{m}] = s[r'_1 \cdots r'_m] $, where $s[q_1 \dots q_p]$ refers to the substring $s_{q_1} \dots s_{q_p}$ and the increasing indices $q_i$'s, $q'_i$'s, $\cdots$, $r_i$'s and $r'_i$'s, and the length of $s$ are given as part of the input. The reconstruction problem consists in finding a string $s$ that verifies the given set of constraints, with a maximum number of distinct letters.  %However, this problem does not admit a similar parametrization than ours, and does not seem to have a polynomial reduction. 
We shall see that these problems are intrinsically different; in particular, the complexity results presented in this article imply the absence of reduction to the reconstruction problem solvable in linear time with respect to the length of the input string $s$. 
In a similar fashion, \emph{De Bruijn} graphs  \cite{de1946combinatorial} are directed graphs representing overlaps between sequences of symbols.
They can be seen as a specialization of the problem in this paper, using window size 2.
%Vertices represent the words one can form over a finite vocabulary, and the oriented arcs relate to a similarity between words (prefix is the same as the suffix). Sequence graphs differ in essence from De Bruijn graphs, as arcs are created by order of appearance of these words in a sequence, instead of the similarity of the given words. Therefore, despite their similarity, the properties of sequences graphs cannot be deduced from the properties of De Bruijn graphs in the context of our study. \todo[inline]{Compare better with $k$-mer De Brujin}
Given a positive integer $k$, the vertices of a De Brujin graph formed from a sequence are the $k$-mers (sub-sequence of $k$ symbols) in the sequence, and edges are formed between a pair $k$-mers when one appears just before another (two adjacent $k$-mers will share $k-1$ consecutive symbols: as a suffix for the first one, and as a prefix for the other). The main difference with our setting is that distinct size-$k$ windows are encoded into distinct nodes even if they share many symbols (in different orders), while our setting has a very sparse set of nodes (one per symbol) and all context information is held by edges. 

Inverse problems studied in the Distance Geometry (DG) literature also bear some similarities. The input of a DG instance consists of a set of pairwise distances between points, having unknown positions in a $d$-dimensional space. A DG problem then consists in determining a set of positions for the points (if they exist), satisfying the distance constraints. Since a position is fully characterized from $d+1$ neighbors, the problem can be solved by finding a sequential order in the points, such that the assignment of a point is always by at least $d+1$ among its neighbors~\cite{liberti2014euclidean} (called  \textit{linear ordering}). Therefore, finding a linear ordering shares some level of similarity with our inverse problems since a realization for a window $w=d+2$  also represents a linear ordering of its nodes, in which $w-1 = d+1$ of the neighbors have lower value with respect to the order.
However, linear ordering in DG to solve our problems is insufficient. First, each element of the sequence $x$ is associated with a unique vertex in the DG instance. In sequence graphs a symbol can be repeated several times, but only one vertex is created in the graph. This implies that the vertex associated to the $i^{\text{th}}$ element $(i \geq w)$  of $x$ can have less than $w-1$ distinct neighbors in its predecessors in $x$. Second, DG graphs are essentially undirected, and loops are not considered, since an element is at distance $0$ from itself. 

\subsection{Paper outline}
In the following section we give a complete overview of our structural and algorithmic results: first for window size 2 (Section~\ref{sec:TwoSeq}), then for larger window sizes (Section~\ref{sec:threeSeq}).
We prove results for window size 2 in  Section~\ref{sec:proofsTwoSeq}, and then give detailed algorithms and hardness reductions for larger windows in Section~\ref{w_3_study}. More precisely, we first focus on the \GU variant (with an XP algorithm and W[1]-hardness proof) in Section~\ref{sec:GU_results}, then hardness proofs for other variants in Section~\ref{sec:nphardnessproofs}, and finally exponential-time algorithms for weighted variants (\GW and \DW) in Section~\ref{sec:effective_general_algo}. Finally, We present a construction yielding exponential-size realizations in Section~\ref{sec:exponential_realizations}.

%\todo[inline]{This inverse problem is somehow similar but not equivalent to the Universal Reconstruction of a String \cite{gawrychowski2020universal}}
%\todo[inline]{Present structure paper}
%After introducing in Section~\ref{sec:defs} the formal definition of a sequence graph and the descriptions of our problems, 

\section{Results overview}\label{sec:pres_results}
In this section, we present our main theoretical results for $w=2$ in Subsection \ref{sec:TwoSeq} and for $w\geq 3$ in Subsection \ref{sec:threeSeq}. For the sake of readability, we postpone full proofs of our results to Sections \ref{sec:proofsTwoSeq} and \ref{w_3_study} respectively. 
Overall, our results reveal a stark contrast between $w=2$, where all relevant problems can be solved in polynomial time, and $w\ge 3$ where most versions of our problems become hard, except for \GU which admits a slicewise polynomial (XP) algorithm parameterized by $w$.

\subsection{Complete characterization of 2-sequence graphs (\texorpdfstring{$w=2$}{w=2})}\label{sec:TwoSeq}

%In this section, we consider $w=2$: algorithm \ref{algo:graph_sequence} adds edges between adjacent elements of the considered sequence.
%\subsection{Unweighted sequence graphs}\label{sec:unweighted_graphs}

A graph has a $2$-realization when there exists a path visiting every vertex and covering all  of its edges (at least once for the unweighted case and exactly $\pi_e$ times for each edge $e$ in the weighted case). This observation enables relatively simple characterization and algorithmic treatment, leading to the results summarized in Table~\ref{table:summary_graph_sequences_2} and Theorem~\ref{thm:2seq_results}. The additional definitions are given below.  

%\todo[inline]{todo: reinclude table}
 \begin{table}[t]
  \caption{Complexity of variants of \NumRealizations{} and \Realizable{} with $w=2$. The possible number of sequences is given for each variant of \NumRealizations{}, and a characterization of yes-instances is given for \Realizable{}.%
  \label{table:summary_graph_sequences_2}
  }
  \small
  %\begin{tabular}{|c|C{5cm}|}  
  {\centering
  %\begin{adjustbox}{{max width=\textwidth}}
    \begin{tabular}{@{}lcccc@{}}
      \toprule
      & \multicolumn{2}{c}{\NumRealizations{2}}&\multicolumn{2}{c}{\Realizable{2}}\\
      Data Instance & Complexity & \#Sequences & Complexity & Characterization\\
      \midrule
      %\tabularnewline
      %\cline{1-3}
      \GU  & P &  $\{0,+\infty\}$ & P &  $G$ is connected \\ %$\psi(G)?$ 
      %\tabularnewline
      % & & \\
      %\tabularnewline
      \GW  & $\#$P-hard & $ \mathbb{N}$ &  P & $\psi(G)$ is (semi-)Eulerian \\
      %\tabularnewline
      %   & &\\
      %\tabularnewline
      \DU & P & $\{0, 1, +\infty \}$  & P & $G$ is a simple step%Theorem \ref{UW_UW2_theorem} % (Cf. Appendix)
      \\
      % \tabularnewline
      %   & & \\
      %\tabularnewline
      \DW  & P & $\mathbb{N} $ (BEST Theorem) & P & $\psi(G)$ is (semi-)Eulerian \\
    \bottomrule
    \end{tabular}
    %\end{adjustbox}\\
   }
 \end{table}

\begin{definition}[$\psi(G)$]
  Let $(G, \Pi)$ be a weighted graph (directed or undirected). $\psi(G)$ is the multigraph with the same vertices as $G$ and with multiplicity  $\pi_{e}$ for each edge $e\in E$.
\end{definition}

%\begin{definition}[Eulerian, semi-Eulerian]\label{DEF:EULER}
%    Let $G$ be a graph (directed or undirected). A cycle in $G$ is {\em Eulerian} if and only if it visits every edge of $G$ exactly once. If $G$ has a Eulerian cycle, it is said to be Eulerian. A path in $G$ is a {\em semi-Eulerian} path if and only if it visits every edge exactly once and its first and last vertices are distinct. If $G$ has a semi-Eulerian path, it is said to be semi-Eulerian. 

    %A graph $G $ (directed or undirected) is Eulerian if and only if there exists a cycle visiting every edge of $G$ exactly once. Such a cycle is said to be Eulerian. $G$ is semi-Eulerian if and only if there exists a path with distinct first and last vertices, and visiting every edge of $G$ exactly once. Such path is said to be semi-Eulerian path.
%\end{definition}

%{\color{red} Alternative definition}
\begin{definition}[(semi-)Eulerian]\label{DEF:EULER}
We say that a path is \emph{(semi-)Eulerian} if it visits all edges of the graph exactly once, and a graph is (semi-)Eulerian if it admits a (semi-)Eulerian path. A (semi-)Eulerian path with identical endpoints is a \emph{Eulerian cycle}, otherwise it is a \emph{semi-Eulerian path}, and this distinction extends to Eulerian and semi-Eulerian graphs  (here this distinction is only made in Proposition~\ref{count_2}).
\end{definition}
\begin{definition}[$R(G)$, $R^{+}(G)$]\label{DEF:3}
  Let $G=(V,E)$ be a digraph.  $R(G)$ is the Directed Acyclic Graph (DAG) such that:  i) every strongly connected components of $G$ is associated to a unique node in $R(G)$, and ii) two strongly connected components  $u \neq v$ in $G$ form an edge $(u,v$) in $R(G)$, provided there exists an edge $(x,y)\in E$ such that $x \in u$ and $y\in v$.
 
 $R^{+}(G)$ is the weighted DAG such that: i) $R^{+}(G)$ has the same vertices and edges as $R(G)$ and ii) the weight of an edge in $R^{+}(G)$ is the number of distinct edges between two strongly connected components in $G$.
\end{definition}

\begin{definition}[simple step graph]\label{def:simple_step}
    Let $G$ be a digraph. $G$ is said to be a {\em simple step} graph (see Figure~\ref{fig:simple_step}) if $R^{+}(G)$ is a directed path and its edges have weight~1.
\end{definition}

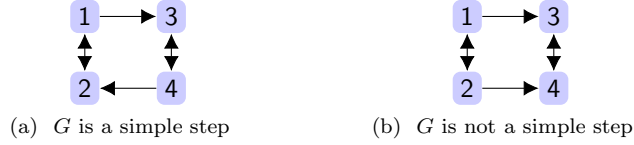
\begin{figure}[!ht]
  %\begin{minipage}[t]{0.45\linewidth}
  \subfloat[ $G$ is a simple step, $R^+(G)$ contains a single node]{ \makebox[0.4\textwidth]{
  %\hspace{1cm}
 
    \begin{tikzpicture}
    %\begin{scope}[minimum width = 4cm]
    \node[main node] (1) {1};
    \node[main node] (2) [below = 0.5cm  of 1] {2};
    \node[main node] (3) [right = 0.75cm  of 1] {3};
  \node[main node] (4) [below = 0.5cm  of 3] {4};
    
    \draw[<->] (1) edge node[right, yshift=5pt] {} (2);
    \draw[->] (1) edge node[above] {} (3);
\draw[<-] (2) edge node[above] {} (4);
  \draw[<->] (3) edge node[above] {} (4);
  
    \node[main node] (1) {1};
    \node[main node] (2) [below = 0.5cm  of 1] {2};
    \node[main node] (3) [right = 0.75cm  of 1] {3};
  \node[main node] (4) [below = 0.5cm  of 3] {4};
    
    \draw[<->] (1) edge node[right, yshift=5pt] {} (2);
    \draw[->] (1) edge node[above] {} (3);
\draw[<-] (2) edge node[above] {} (4);
  \draw[<->] (3) edge node[above] {} (4);

            \begin{scope}[shift={(2.5cm,-0.5cm)}]
                \node [reduced node] () {1234};                   
                
            \end{scope}
  
    %\end{scope}
    \end{tikzpicture}
  %\end{minipage}
  }
 }\hfill
 \subfloat[ $G$ is not a simple step, $R^+(G)$ is a path with a weight-2 edge]{ \makebox[0.4\textwidth]{
  %\hspace{1cm}
    \begin{tikzpicture}
    %\begin{scope}[minimum width = 4cm]
    \node[main node] (1) {1};
    \node[main node] (2) [below = 0.5cm  of 1] {2};
    \node[main node] (3) [right = 0.75cm  of 1] {3};
      \node[main node] (4) [below= 0.5cm  of 3] {4};
    
    \draw[<->] (1) edge node[right, yshift=5pt] {} (2);
    \draw[->] (1) edge node[above] {} (3);

  \draw[->] (2) edge node[above] {} (4);
  \draw[<->] (3) edge node[above] {} (4);
  
            \begin{scope}[shift={(2.5cm,-.5cm)}]
                \node [reduced node] (12) {12};       
                \node [reduced node] (34) [right = 0.75cm  of 12] {34};       
          \draw[rededge, ->] (12) edge node[below, black] {2} (34);         
                
          %  \end{scope}
    \end{scope}
    \end{tikzpicture}
  %\end{minipage}
  }
 }
 \caption{Illustration of Definition~\ref{DEF:3} of graph $R^+(G)$ (drawn with brown edges and background) and Definition~\ref{def:simple_step} of a simple step graph. The graph in (a) has a single strongly connected component (yielding a trivial graph $R^+(G)$). The graph in (b) has two strongly connected components connected with two distinct edges ((1,3) and (2,4)), so they form  a weight-2 edge in $R^+(G)$.\label{fig:simple_step}}
 \end{figure}
\begin{theorem}\label{thm:2seq_results}

    A weighted graph $G$ (directed or undirected) admits a 2-realization if and only if $\psi(G)$ is (semi-)Eulerian. 
    An unweighted undirected (resp. unweighted directed) graph admits a 2-realization if and only if it is connected (resp. simple step). 
    
All variants of \Realizable{2} are in P. For  \NumRealizations{2}, the \GW variant is $\#$P-hard, and all others are in P.
\end{theorem}

\subsection{General complexity and algorithmic results (\texorpdfstring{$w \geq 3$}{w>2})}\label{sec:threeSeq}

In this subsection we present the remaining complexity results, which are summarized in Theorem \ref{theoremsumup} and Table \ref{table:summary_graph_sequences}. We first show that \GU-\Realizable{w} $\in P$ for any integer $w \geq 3$. Besides, for \GU, the number of realizations of a graph $G$ is either $0$ (not realizable), $1$ or $+\infty$ (realizable in both cases).
These three cases can be tested in polynomial time using our algorithm (presented in Section \ref{w_3_study}), showing that \GU-\NumRealizations{w} $\in P$, for any integer $ w \geq 3$. All proofs of the following statements are given in Section \ref{w_3_study}.

\begin{theorem}
\label{theoremsumup}For any integer $w\geq 3$, the \GW, \DU and \DW variations of \NumRealizations{w} and \Realizable{w} are NP-hard, and the \GU variations are in P. 
\end{theorem}
In other words, \NumRealizations{} and \Realizable{} parameterized by $w$ are para-NP-hard in the \GW, \DU and \DW variations and slicewise polynomial (XP) in the \GU variations. 
%%todo reinsert proof
 \begin{table}
  \caption{Complexity for various instances of our problems $(w \geq 3)$. We remind that a para-NP-hard problem does not admit any XP algorithm unless P=NP.}
  \label{table:summary_graph_sequences}
  \small
  %\begin{tabular}{|c|C{5cm}|}
  %
  %{\centering
  %\begin{adjustbox}{max width=\textwidth}
    \begin{tabular}{lcccc}
      \toprule
      & \multicolumn{2}{c}{Constant $w$, $w\geq 3$}
      & \multicolumn{2}{c}{Parameter $w$}  \\
      & \multicolumn{1}{c}{\NumRealizations{w}}&\multicolumn{1}{c}{\Realizable{w}} & \NumRealizations{} & \Realizable{} \\
    
      Variation & Complexity & Complexity & Complexity & Complexity\\
    
      \midrule
      %\tabularnewline
      %\cline{1-3}
      \GU  & P  & P   &  W[1]-hard; XP & W[1]-hard; XP \\ %$\psi(G)?$ 
      %\tabularnewline
      % & & \\
      %\tabularnewline
      \GW  & NP-hard  & NP-hard   &  para-NP-hard  &  para-NP-hard \\
      %\tabularnewline
      %   & &\\
      %\tabularnewline
      \DU & NP-hard   & NP-hard  & para-NP-hard  & para-NP-hard \\
      % \tabularnewline
      %   & & \\
      %\tabularnewline
      \DW & NP-hard  & NP-hard & para-NP-hard  & para-NP-hard
      %\DW  & Open &  Open  & Open & Open
    \end{tabular}
  %\end{adjustbox}
  %}
 \end{table}

We further investigate the parameterized complexity for \GU-\Realizable{}, since an XP algorithm, here in time $O(n^{w})$, can sometimes be improved into an FPT algorithm running in time $f(w)n^{O(1)}$. This is however, not the case here (under usual complexity assumptions), as we prove the parameterized hardness of this problem.
\begin{restatable}{theorem}{theoremcliquereduction}\label{theorem:cliquereduction}
{\GU}-\RealizableG{}   is {\normalfont W}$[1]$-hard for parameter $w$.
\end{restatable}

Finally, we describe in Section~\ref{sec:effective_general_algo} an integer linear formulation of \DW and \GW-\Realizable{}, and a $\mathcal{O}(n^w 2^{w\,p})$ dynamic programming algorithm for \NumRealizations{}.

\section{The special case of window size 2 
(\texorpdfstring{$w=2$}{w=2})}\label{sec:proofsTwoSeq}

In this section we present the proofs of the results gathered in Table \ref{table:summary_graph_sequences_2} and Theorem~\ref{thm:2seq_results}. Apart from the \GU variant (Lemma~\ref{results_unweighted_undirected}), we use direct reductions to standard well-known problems in graph theory. The \DU variant can be treated with a reduction to simple step graphs (cf. Definition \ref{def:simple_step}, Lemma~\ref{DU_2_characterization} and Corollary~\ref{corollary:DU_2_count}). The weighted cases (\GW and \DW) are treated with direct reductions to the problem of existence and counting of  (semi-)Eulerian paths in a graph (Lemma~\ref{lemma:count:2seq}).
%\YCom{Yann: Broken link to results}

\subsection{Unweighted realizability in undirected (\GU) and directed (\DU) graphs}

The following two propositions are obtained using a simple greedy algorithm: starting from any vertex of the graph, if an edge is not realized by a candidate sequence, append any path from the current last vertex of the sequence to the missing edge, and repeat until all edges are realized. 

\begin{lemma}[\GU characterization]\label{results_unweighted_undirected}
  If $G=(V,E)$ is unweighted and undirected, with $|V|>1$, the following are equivalent:\\
  \statement{(i)} $G$ is connected\\
  \statement{(ii)} $G$ has a $2$-realization\\
  \statement{(iii)} $G$ admits an infinite number of $2$-realizations. 
  
  In these conditions, a $2$-realization can start and end at any vertex.
\end{lemma}
\begin{proof}
    (i)$\Rightarrow$ (ii) is obtained with the greedy algorithm described above. For (ii)$\Rightarrow$(iii), simply notice that any 2-realisation ending with $ab$ can be continued with arbitrarily many additional repeats of $ab$. For (iii)$\Rightarrow$(i), a path between any two vertices can be found using any substring of any realization starting and ending with the two given vertices.
\end{proof}

The previous characterization does not extend to strongly connected digraphs. 
However, we do get the following sufficient conditions, which are not necessary, as seen in Figures~\ref{fig:counter_example_2} and~\ref{fig:counter_example_DIG_EULER}. Furthermore, they only apply in the unweighted case, as shown in Figure~\ref{counter_ex_weighted_w2}. 

\begin{proposition}\label{UW_UW2}\label{prop:sufficient_condititions_DU2}
  Let $G=(V,E)$ be an unweighted digraph. 
 
 (i) If $G$ is strongly connected then $G$ has a $2$-realization. A $2$-realization can start or end at any given vertex of $G$.
 
 (ii) If $G$ is (semi-)Eulerian then $G$ has a $2$-realization.

\label{UW_UW2b} 
\end{proposition}
\begin{proof}
    Condition (i) is obtained  using the greedy algorithm described above in the proof of Lemma \ref{results_unweighted_undirected}.
    
    Condition (ii) follows from the definition: a (semi-)Eulerian path is a 2-realization in the context of unweighted graphs.
\end{proof}

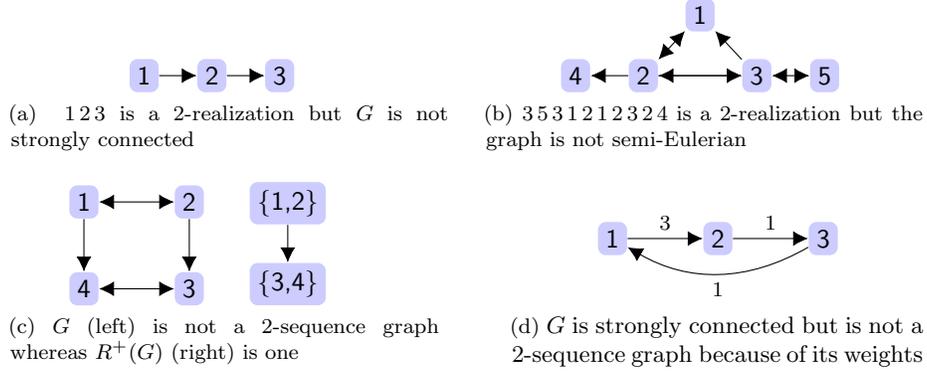
\begin{figure}[!t]
  %\begin{minipage}[t]{0.45\linewidth}
  \subfloat[ $1 \, 2 \, 3$ is a $2$-realization but $G$ is not strongly connected]{ \makebox[0.45\textwidth]{
  \hspace{1cm}
    \begin{tikzpicture}
    \begin{scope}[minimum width = 3cm]
    \node[main node] (1) {1};
    \node[main node] (2) [right = 0.5cm  of 1] {2};
    \node[main node] (3) [right = 0.5cm  of 2] {3};

    \draw[->] (1) edge node[right, yshift=5pt] {} (2);
    \draw[->] (2) edge node[above] {} (3);
    
    \end{scope}
    \end{tikzpicture}
  %\end{minipage}
  }
    \label{fig:counter_example_2}
  }
  \hfill
  \subfloat[$3 \, 5 \, 3 \, 1 \, 2 \, 1 \, 2 \,  3 \,  2 \, 4$ is a $2$-realization but the graph is not semi-Eulerian]{
  \makebox[0.45\textwidth]{
  %\begin{minipage}[t]{0.45\linewidth}
    %\begin{minipage}[c]{0.3\linewidth}
    \begin{tikzpicture}
    \begin{scope}[xshift=4cm]
    \node[main node] (1) {1};
    \node[main node] (2) [below left = 0.5cm  of 1] {2};
    \node[main node] (3) [below right = 0.5cm  of 1] {3};
    \node[main node] (4) [left = 0.5cm  of 2] {4};
    \node[main node] (5) [right = 0.5cm  of 3] {5};
    
    \draw[<-] (1) edge node[right, yshift=5pt] {} (3);
    \draw[<->] (1) edge node[right, yshift=5pt] {} (2);
    \draw[<->] (2) edge node[right, yshift=5pt] {} (3);
    \draw[->] (2) edge node[right, yshift=5pt] {} (4);
    \draw[<->] (3) edge node[right, yshift=5pt] {} (5);
    
    \draw[->] (2) edge node[above] {} (3);
    
    \end{scope}
    \end{tikzpicture}
    }
    
  \label{fig:counter_example_DIG_EULER}
    }

  \subfloat[\small $G$ is strongly connected but is not a $2$-sequence graph because of its weights \label{counter_ex_weighted_w2}]{\makebox[0.45\textwidth]{ 
  \begin{tikzpicture}
  \begin{scope}[xshift=4cm]
  \node[main node] (1) {1};
  \node[main node] (2) [right = 1cm  of 1] {2};
  \node[main node] (3) [right = 1cm  of 2] {3};
  
  \draw[->] (1) edge node[above] {{\footnotesize $3$}} (2);
  \draw[->] (2) edge  node[above] {{\footnotesize $1$}} (3);
  \draw[->] (3) edge[bend left=30] node[below] {{\footnotesize $1$}} (1);
  \end{scope}
  \end{tikzpicture}
  }}
    \hfill
  \subfloat[$G$ (left) is not a $2$-sequence graph whereas $R^{+}(G)$ (right) is one]{\makebox[0.45\textwidth]{ 
      \begin{tikzpicture}
      \begin{scope}[xshift=4cm]
      \node[main node] (1) {1};
      \node[main node] (2) [right= 1cm  of 1] {2};
      \node[main node] (3) [below = 0.7cm  of 2] {3};
      \node[main node] (4) [below = 0.7cm  of 1] {4};
      \draw[<->] (1) edge node[right] {} (2);
      \draw[->] (2) edge node[below] {} (3);
      \draw[<->] (3) edge node[right] {} (4);
      \draw[->] (1) edge node[right] {} (4);
      \end{scope}
      \end{tikzpicture}
   
      \begin{tikzpicture}
      \begin{scope}[minimum width=2cm]%[xshift=6cm]
      \node[reduced node] (c1){12};
      \node[reduced node] (2) [below= 0.5cm  of c1] {34};
      \draw[rededge, <-] (2) edge[] node[] {} (c1);
      \end{scope}
      \end{tikzpicture}
    }
    \label{counter_example_dw_2first}
    }
   
\caption{Some special cases for $w=2$, acting as counterexamples for variants of Propositions~\ref{prop:sufficient_condititions_DU2} and~\ref{prop:necessary_condititions_DU2}.   }
\end{figure}
For the special case of DAGs, we have the following simple characterization of 2-sequence graphs.
\begin{proposition}\label{prop:DU_DAG}
    If $G=(V,E)$ is a DAG, then it is a $2$-sequence graph if and only if it is a directed path, in which case $G$ has a unique 2-realization.
\end{proposition}
\begin{proof}
The backward direction ($\Leftarrow$) is an application of Proposition~\ref{prop:sufficient_condititions_DU2}(ii): directed paths admit a unique semi-Eulerian path, that yield a unique 2-realization.

For ($\Rightarrow$), note first that a DAG is a path if and only if it is connected with in- and out-degree at most 1. 
 Let us suppose $G$ is a $2$-sequence graph and is not a directed path, then we show that it contains a cycle. Indeed,  there exists a vertex $v$ having either two children, or two parents. In the first case, denote $c_1$ and $c_2$ for the two distinct children of $v$. 
 Then there exists a walk going through both $(v, c_1)$ and $(v, c_2)$ (w.l.o.g. in this order), so $G$ also has a path from $c_1$ to $v$, which creates a cycle with $(s,c_1)$. Similarly, if $v$ has two parents $p_1$ and $p_2$, then any 2-realization yields a walk through both $(p_1,v)$ and $(p_2,v)$, and thus a cycle. 
\end{proof}

We now turn to general digraphs and give a necessary conditions for $G$ to admit a 2-realization (which is not sufficient, as shown using an example in Figures~\ref{counter_example_dw_2first}).

\begin{proposition}\label{prop:necessary_condititions_DU2}
 Let $G=(V,E)$ be a digraph. If $G$ is a $2$-sequence graph then $R(G)$ is a $2$-sequence graph.  
\end{proposition}
\begin{proof}
 Let $G$ be a $2$-sequence graph, and for the sake of contradiction let us suppose that $R(G)$ is not a $2$-sequence graph. Since $R(G)$ is a (weakly) connected DAG, then using Proposition~\ref{prop:DU_DAG}, it cannot be a directed path, so $R(G)$ has either a node having two children or two parents. Let us assume  without loss of generality that $R(G)$ has a node $v$ in $R(G)$ with two children $c_1$ and $c_2$. Recall that $v$, $c_1$ and $c_2$ represent three distinct connected components of $G$ (in particular three sets of vertices in $G$ that have empty intersection). Hence, there exist $v_1, v_2 \in V, w_1 \in c_1$, and $w_2 \in C_2$ such that $(v_1,w_1) \in E$ and $(v_2,w_2) \in E$. 
    Consider now the 2-realization of $G$,  assuming without loss of generality that $(v_1,w_1)$ is realized before $(v_2,w_2)$. Then there exists a path between $w_1$ and $v_2$ in $G$, which implies that $w_1$ belongs to the same component as $v_1$ and $v_2$: a contradiction.
\end{proof}

Based on the sufficient and necessary conditions above, we can now converge to a complete characterization of 2-sequence graphs in the \DU setting:
\begin{lemma}[\DU characterization]\label{DU_2_characterization}
  Let $G=(V,E)$ be an unweighted digraph. $G$ is a $2$-sequence graph if and only if it is a simple step graph. This property can be verified in linear time.
\end{lemma}

\begin{proof}%[Proof (Th.~\ref{UW_UW2_theorem})]
  If $G$ is a $2$-sequence graph,  $R(G)$ is a $2$-sequence graph using Proposition~\ref{prop:necessary_condititions_DU2}. Proposition~\ref{prop:DU_DAG} implies that $R(G)$ and $R^{+}(G)$ are directed paths. Moreover, if $R^{+}(G)$ has an edge with weight greater that $1$, then there would be more than one edge between two strongly connected components $c_1$ and $c_2$. All these edges go in the same direction otherwise $c_1 \cup c_2$ would form a strongly connected component. This is a contradiction since any $2$-realization would have to go from $c_1$ to $c_2$ and then come back to $c_1$ (or conversely), which would make $c_1 \cup c_2$ a strongly connected component.

  Conversely, let us suppose $R^{+}(G)$ is a directed path and its weights are equal to one. By definition, there exists a list of sets of vertices $\mathcal{P}=(c_1, ..., c_p)$ such that:  
  \begin{itemize}
  \item[(i)] the entries of $\mathcal{P}$ form a partition of $V(G)$, i.e. $c_i \subset V(G)$ with $\lvert c_i\rvert \geq 1$, $\bigcup_{i \in \{1, \cdots, p\}} c_i = V(G) $ and for any $i \neq j$, $c_i \cap c_j = \emptyset$.
  
   \item[(ii)] For any $i\in \{1, \cdots, p-1\}$, there exists a unique 
  edge $(u,v)$ of $G$ with $u\in c_i$ and $v\in c_{i+1}$.
    \end{itemize} We construct a $2$-realization $y$ for $G$ by means of the following procedure.

  Base case:  $c_1$ is a strongly connected component of $G$, we initialize $y$ with any $2$-realizations of $c_1$ (which exists by Proposition~\ref{prop:sufficient_condititions_DU2}(i)).

   For $i \in \{1,..,p-1\}$: let $e=(v, w)$ be the single edge with $v\in c_i$ and $w\in c_{i+1}$. By construction, all the edges induced by $c_i$ have already been added to $y$. Suppose at the previous step the last vertex added is $z \in c_i$. % In case $z\neq v$, 
    We first add all vertices of a walk starting at $z$ (excluded) and ending at $v$. Then, consider a walk starting at $w$ (included) and which visits every edge of $c_{i+1}$ (again using  Proposition~\ref{prop:sufficient_condititions_DU2}(i)). We add all vertices of this walk after $v$.   
    
 The process stops when $i=p-1$, and all edges of $G$ are realized by $y$.
 
Finally, note that verifying if a graph is simple step is linear, as it directly follows from a strongly connected components decomposition. 
\end{proof}

A direct consequence of Lemma~\ref{DU_2_characterization} is the following:

\begin{corollary}\label{corollary:DU_2_count}
  Let $G$ be an unweighted digraph. The possible numbers of $2$-realizations for $G$ are only 0, 1 and $+\infty$. Moreover, $G$ admits a unique  $2$-realization if and only if $G$ is a directed path.
\end{corollary}
\begin{proof}
First, if $G$ is a DAG, then by Proposition~\ref{prop:DU_DAG} it has either zero or a unique $2$-realization (exactly one if it is a directed path). If $G$ is not a DAG, $G$ has a cycle $u_0 u_1 \ldots (u_\ell = u_0)$ (possibly with $\ell=1$ in case of self-loops) and admits a $2$-realization~$y$. Then $y$ has at least one occurrence of $u_0$, and a strictly longer 2-realization $y'$ can be obtained by inserting $u_0 \ldots u_{\ell-1}$ just before any occurrence of $u_0$. Therefore $G$ has infinitely many $2$-realizations.
\end{proof}

\subsection{Weighted realizability in undirected (\GW) and directed (\DW) graphs}\label{sec:2:weighted}

The weighted cases (\GW and \DW) cannot be treated similarly due to the weight constraints implying that a weighted graph has a finite number of realizations (as was seen in   Figure~\ref{counter_ex_weighted_w2}). However, in this setting, we can use (semi-)Eulerian paths to obtain the desired result.

\begin{theorem}\label{results_weighted_2}
  If $G$ is a weighted graph (possibly directed), with weight matrix $\Pi(G)$, then: $G$ is 2-realizable if and only if $\psi(G)$ is connected and (semi-)Eulerian. 
\end{theorem}

%\todo{define eulerian and semi-eulerian}

This theorem follows from the following stronger result, that also relates the number of 2-realizations to the number of (semi-)Eulerian paths of $\psi(G)$.

\begin{lemma}\label{lemma:count:2seq}
    Let $G=((V,E), \Pi)$ be a weighted $2$-sequence graph  (possibly directed). Let $\mathcal{E}$ be the set of (semi-)Eulerian paths of $\psi(G)$ and $\mathcal{S}$ be the set of $2$-realizations of $G$. Then
    $$
                |\mathcal{E}| = | \mathcal{S} | \prod_{e\in E}{\pi_e !} 
   $$
\end{lemma}

\begin{proof}
First note that (semi-)Eulerian paths of $\psi(G)$ (writing $h$ for the number of edges in $\psi(G)$) can be characterized by a pair $(u_0 u_1 \ldots u_h, e_1\ldots e_h)$ where each $u_i$ is a vertex of $G$, $e_1\ldots e_h$ is a permutation  of the edges of $\psi(G)$, and $e_i=(u_{i-1},u_i)$ (directed case) or $e_i=\{u_{i-1},u_i\}$ (undirected case). Note that  $u_0 u_1 \ldots u_h$ is a 2-realization of $G$, and that, conversely, a (semi-)Eulerian path can be obtained from any $u_0 u_1 \ldots u_h$ by taking $e_i$ to be one copy of $(u_{i-1},u_i)$ or $e_i=\{u_{i-1},u_i\}$ for each $i$ (the path indeed goes through all $\pi_{e}$ copies of each edge $e$ in $\psi(G)$ by definition of weighted 2-realizations).

Consider the map:  
\begin{equation}
            \begin{alignedat}{2}
    f: \mathcal{E} & \longrightarrow \mathcal{S}       \\
    (u_0 u_1 \ldots u_h, e_1\ldots e_h) & \mapsto   (u_0 u_1 \ldots u_h)
    \end{alignedat}
    \end{equation}
     We have already noted that $f$ is surjective  (visiting multiple copies of the same edge in different orders give the same 2-realization but with different (semi-)Eulerian paths).
   An element $x \in \mathcal{E}$ can be seen as a list of edges of $G$, each appearing $\pi_e$ times, since each edge  $\psi(G)$ is obtained by copying $\pi_e$ times every edge of $G$. Therefore this map is not injective, as soon as there is one $\pi_e > 1$, because one can permute the corresponding edges in the (semi-)Eulerian path, and the corresponding $2$-sequence is the same.
   
    We thus consider the following relation $\sim$ on $\mathcal{E}$: For two (semi-)Eulerian paths $P_1$ and $P_2$, $P_1 \sim P_2 \iff$ $P_1$ can be obtained from $P_2$ by permuting edges of $\psi(G)$ that are copies of the same edge in $G$. $\sim$ is an equivalence relation because it is symmetric, transitive and reflexive.  Let $\mathcal{E}/\sim$ be $\mathcal{E}$ quotiented by $\sim$. We have $ P_1 \sim P_2 \iff f(P_1) = f(P_2) $ (equivalently,  $P_1$ and $P_2$ yield the same sequence of vertices), so $|\mathcal{S}|$ is the number of equivalence classes of $\sim$, or equivalently, $|\mathcal{E}/\sim|$. Note that each equivalence class of $\sim$ has cardinality $\prod_{e\in E}\pi_e !$ (number of permutations which are product of permutations with disjoint supports, where each support has size $\pi_e$).
  Therefore $|\mathcal S| = |\mathcal{E}/\sim| = |\mathcal{E}| (\prod_{e\in E}\pi_e)^{-1}$.
\end{proof}

On the one hand, counting the number of (semi-)Eulerian paths in a undirected graph is a  $\#P$-complete problem \cite{brightwell2005counting}. Since $G \mapsto \psi(G)$ is bijective, counting the number of $2$-realizations is also $\#P$-complete in the \GW setting. On the other hand, for the \DW setting, counting (semi-)Eulerian paths of a weighted digraph is in P, and can be derived using the following proposition (writing $\deg_{\psi(G)}(v)$ for the  indegree of a vertex $v$ in $\psi(G)$, i.e.  $ \deg_{\psi(G)}(v) = \sum_{\substack{u \in V }}\pi_{(u,v)}  $): 

\begin{proposition}\label{count_2}
  Let $G=(V,E)$ be a weighted digraph, with $\Pi(G)$ an $n\times n$ matrix of integers. Then, the number $p_2$ of $2$-realizations is given by
  \begin{align}\label{formula_2:count}
  &\text{-If $\psi(G)$ is Eulerian,} \qquad p_2 =  \frac{t(\psi(G))}{\prod_{e\in E} \pi_e! } \prod_{v\in V} \bigl(\deg_{\psi(G)}(\psi(v))-1\bigr)!
 \end{align}
 where $t(G)$ is the number of spanning trees of a graph $G$. If $L$ is the Laplacian matrix of $G$ and $Sp(L)$ the set of eigenvalues of $L$, then %$t(G)$ can be expressed as
  \begin{equation*}%\label{formula_matrix_tree}
  t(G)=\prod_{\substack{\lambda_i \in Sp(L) \\ \lambda_i \neq 0}} \lambda_i
  \end{equation*}
 
    - If $\psi(G)$ is semi-Eulerian, make it Eulerian by adding one arc $(u,v)$ between the two vertices with unbalanced degrees ($u$ is the one with the least outdegree, $v$ has the least indegree). Then apply Formula \ref{formula_2:count} to $\tilde{\psi}(G):=\psi(G) + (u,v)$, and divide the output by the number of vertices $\lvert V \rvert $.

  %~\ref{formula_matrix_free}.
\end{proposition}
\begin{proof}
  The case of $\psi(G)$ being Eulerian is a direct consequence of Lemma \ref{lemma:count:2seq}, BEST Theorem~\cite{de1951circuits} and Matrix Tree Theorem~\cite{chaiken1982combinatorial}.

 When $\psi(G)$ is semi-Eulerian, this follows from the fact that $\psi(G)$ is semi-Eulerian if and only if $\psi(G) + (u,v)$ is Eulerian where: $u$ is the the vertex whose outdegree is less than its indegree, and $v$ is the vertex whose indegree is less than its outdegree. In that case, the number of semi-Eulerian paths of $\psi(G)$ is exactly the number of Eulerian paths of $\psi(G)+(u,v)$ divided by $\lvert \psi(G) \rvert = \lvert V \rvert$ (since for one semi-Eulerian path in $\psi(G)$ there are exactly $\lvert V \rvert $ Eulerian paths in $\psi(G) + (u,v)$).
  \end{proof}

\begin{corollary}\label{corollary:DW_GW_2_count}
  Let $G$ be a weighted graph (directed or undirected). For every non-negative integer $n$, there exists a sequence graph having $n$ $2$-realizations. 
\end{corollary}
\begin{proof}
    Let $n \geq 0$ be an integer. The case $n=0$  and $n =1$ are trivial, so we suppose now that $n > 1$. In the directed case, simply consider the oriented cycle  $C_n$ on $n$ vertices where all edges have weight $1$. Then, $C_n$ has exactly $n$ realizations where each sequence is determined by one of the $n$ starting vertices.

    In the undirected case, if $n$ is even, the (undirected) cycle with unit weights $C_{\frac{n}{2}}$  gives $n$ $2$-realizations (both directions are now allowed). If $n=2p+1$ with $p \in \mathbb{N}$, consider  the  sequences defined as follows and their reverse:
    \begin{align*}
        & v_1 v_2   \cdots v_{p} \, x x \, v_{p} \cdots v_2 v_1 &(1)\\
      &  v_{i+1}  v_i \cdots v_2 v_1 v_2  \cdots  v_p \, x \, x  \,  v_{p} \cdots v_{i+1}  \quad \forall i \in [p-1] &(2)\\
      & x \, v_p v_{p-1} \cdots v_2 v_1 v_2  \cdots  v_p \, x \, x  &(3) 
    \end{align*}
    This represents a total of $1+2(p-1)+2= n$ sequences (since $(1)$ is its own reverse). 
    First observe that all these sequences yield the same sequence graph $G$, shown below
    \begin{center}
    \begin{tikzpicture}
  \begin{scope}[xshift=4cm]
    \node[] at (-1,0) {$G$};
  \node[main node] (1) {$x$};
  \node[main node] (2) [right = 1cm  of 1] {$v_p$};
  \node[main node] (3) [right = 1cm  of 2] {$v_{p-1}$};
\node[main node] (4) [right = 1cm  of 3] {...};
\node[main node] (5) [right = 1cm  of 4] {$v_1$};
   \draw [->] (1) edge[loop above] node{{\footnotesize $1$}} (1); 
  \draw[-] (1) edge node[above] {{\footnotesize $2$}} (2);
  \draw[-] (3) edge node[above]{{\footnotesize $2$}} (2);
\draw[-] (4) edge node[above]{{\footnotesize $2$}} (3);
\draw[-] (5) edge node[above]{{\footnotesize $2$}} (4);
  \end{scope}
  \end{tikzpicture}
\end{center}

Thus $G$ admits at least $n$ $2$-realizations. It remains to show that every $2$-realization of $G$ is of one of the forms $(1)$, $(2)$, $(3)$ above. This is based on the observation that the vertex $x$ must appear twice or three times. 
    
    If $x$ appears exactly twice, it has to appear next to itself and  next to $v_p$ twice, so the sequence contains the subsequence $v_p \, x \, x \, v_p$.  We are exactly in case $(1)$ or $(2)$ as the rest of the sequence has to verify the adjacency constraints between the $v_i$'s.
        
    If $x$ appears exactly three times, then it has one occurrence as the first element and another as the last element. Moreover,  it has to appear next to itself as $x \, x$ (either for the first or last occurrence).  The rest of the sequence is then entirely determined by the adjacency constraints and we are in case $(3)$.
\end{proof}

\section{Complexity and algorithms for general window sizes (\texorpdfstring{$w\geq 3$}{w>=3})}\label{w_3_study}

The characterization of general sequence graphs differs from the one of $2$-sequence graphs: the undirected graph in Figure~\ref{fig:counter_example_undirected_3}
satisfies the conditions of Lemma~\ref{results_unweighted_undirected} but has no 3-realization,  and similarly in a directed setting the graph in Figure~\ref{fig:counter_example_directed_3} satisfies the conditions of Lemma~\ref{lemmaDUrealizable} but is not 3-realizable. 

In fact, there is no simple characterization of realizable graphs for larger window size, and most variants of \Realizable{} are already NP-hard for window size 3 (Section~\ref{sec:nphardnessproofs}). However, we do get a polynomial-time algorithm for the \GU variant for any fixed window size that we present first in Section~\ref{sec:GU_XP_alg}, matched with a parameterized complexity lower bound shown in Section~\ref{sec:GU_W1_hard}.

\begin{figure}[ht]
  \centering
  \subfloat[\small $G$ is connected but not a $3$-sequence graph\label{fig:counter_example_undirected_3}]{
    %\begin{minipage}[c]{0.3\linewidth}
    \begin{tikzpicture}
      \begin{scope}%[minimum width = 4cm]%[xshift=4cm]
        \node[main node] (1) {1};
        \node[main node] (2) [right = 1cm  of 1] {2};
        \node[main node] (3) [right = 1cm  of 2] {3};
        \draw[-] (1) edge node[right] {} (2);
        \draw[-] (2) edge node[below] {} (3);
      \end{scope}
    \end{tikzpicture}
  }\qquad \qquad
  \subfloat[\small $G$ is strongly connected but is not a $3$-sequence graph \label{fig:counter_example_directed_3}]{
    \centering
    \begin{tikzpicture}
      \begin{scope}%[minimum width = 4cm]
        \node[main node] (1) {1};
        \node[main node] (2) [right = 1cm  of 1] {2};
        \node[main node] (3) [right = 1cm  of 2] {3};
        \draw[<->] (1) edge node[right] {} (2);
        \draw[<->] (2) edge node[below] {} (3);
      \end{scope}
    \end{tikzpicture}
  }
  %\captionsetup{justification=centering}
  \caption{Non-realizable graphs for window size 3. Indeed, any 3-realization of length at least 3 in an undirected (resp. directed ) graph yields either a self-edge (resp. self-loop) or a clique (resp. tournament) of size $3$, and these graphs have neither.} 
  
\end{figure}

\subsection{Parameterized results for undirected unweighted graphs (\GU)}
\label{sec:GU_results}
For undirected unweighted graphs (\GU), we describe the construction of a size-$n^w$ auxiliary graph reducing the question of realizability to mere connexity, thus giving a slicewise polynomial (XP) algorithm for the $w$ parameter. We then show that such $n^w$ factor in the complexity is unavoidable due to a W[1]-hardness reduction from \Clique.

\subsubsection{Slicewise polynomial (XP) algorithm for {\GU-\Realizable{}}}
\label{sec:GU_XP_alg}

We introduce the following auxiliary graph used in our polynomial time algorithm for \GU-\Realizable{w}.
See Figure~\ref{fig:illustration_def} for an example.
\begin{definition}\label{def:aux_graph_H}
   Let $G=(V,E)$ be an undirected graph and $k\geq 2$. We write $v_{1:k}$ as a shorthand for a length-$k$ string of nodes $v_1\ldots v_k$.
   Let $H^{k}(G)=(V^{(k)},E^{(k)})$ be the undirected graph with
  \begin{align*}
V^{(k)}& = \{v_{1:k} \mid \forall 1\leq i<j\leq k, \{v_i,v_j\}\in E \}\}\\
       E^{(k)}&=\{\{u_{1:k},v_{1:k}\}\mid u_{2:k} = v_{1:k{-}1}\text{ and } \{u_1, v_k\}\in E\}   
  \end{align*}
\end{definition}

{
  \begin{figure}[H]
    \captionsetup[subfigure]{justification=centering}
    	\centering
%	 \begin{subfigure}[h]{0.45\textwidth}
\centering
\subfloat[$G$]
{\begin{tikzpicture}
	\begin{scope}[xshift=4cm]
	\node[main node] (1) {1};
	\node[main node] (2) at (0,1) {2};
	\node[main node] (3) at (0,-1) {3};
	\draw [ -] (1) edge[out=10, in=-10, looseness=15] (1);
	\draw[-] (1) edge node[right] {} (2);
	\draw[-] (3) edge node[right]{} (1);	
	\end{scope}
	\end{tikzpicture}}\qquad \qquad
\subfloat[$H^{(2)}$]
{\begin{tikzpicture}
	\begin{scope}[xshift=4cm]
	\node[main node]  (11) at (0,0) {11};	
	\node[main node]  (13) at (-.5,-1) {13};
	\node[main node]  (31) at (.5,-1) {31};
		\node[main node]  (21) at (.5,1) {21};
	\node[main node]  (12) at (-.5,1	) {12};
	
	\draw[red, -] (11) edge node[right] {} (13);
	\draw[red, -] (13) edge node[right] {} (31);
	\draw[red, -] (31) edge node[right] {} (11);
	\draw[red, -] (11) edge node[right] {} (21);
%	\draw[red, -] (21) edge node[right] {} (12);	
	\draw[red, -] (11) -- (12); 
    \draw[red, -] (12) -- (21);
    \draw [red, -] (11) edge[out=10, in=-10, looseness=15] (11);
	%\draw[red, -] (43) -- (32); 
	%\draw[red, -] (32) -- (24); 
	
	\end{scope}
	\end{tikzpicture}}
	\qquad \qquad
%\subcaptionbox{DAG $R(H^{(1)})$}[0.33\textwidth]
\subfloat[$H^{(3)}$]
{\begin{tikzpicture}
	\begin{scope}[xshift=4cm]

	\node[main node]  (131) at (0,-1) {131};
	\node[main node]  (311) at (-1,-.5) {311};
    \node[main node] (113) at (1,-.5) {113};
	\node[main node]  (111) at (0,0) {111};
    \node[main node] (211) at (-1,.5) {211};

	\node[main node]  (121) at (0,1) {121};
    \node[main node] (112) at (1,.5) {112};

	\draw[red, -] (113) edge (131);
	
    \draw[red, -] (311) edge (111);

    \draw[red, -] (111) edge (113);

    \draw[red, -] (111) edge (112);

    \draw[red, -] (211) edge (111);

    \draw[red, -] (121) edge (211);

    \draw[red, -] (112) edge (121);

    \draw[red, -] (131) edge (311);

    \draw [red, -] (111) edge[out=10, in=-10, looseness=15] (111);
 
	\end{scope}
	\end{tikzpicture}}

\subfloat[A walk visiting all arcs of (the single connected component of) $H^{(3)}$ and its underlying string that is a 4-realization of $G$ by Lemma~\ref{lem:GU_reverse}.]
{\begin{tikzpicture}
		\begin{scope}[shift={(-5.5,0)}]
		\foreach \u [count=\i] in {111,211,121,112,111,111,113,131,311,111} {
			\node[main node] (y\i) at (\i*.9,0) {\u};
		}
	   \foreach \v [count=\u] in {2,...,10} {
	    \draw[red] (y\u) -- (y\v);
	}
\end{scope}
\begin{scope}[shift={(-4.2,0)}]	
\foreach \i/\off [count=\u] in {1/0,4/0,7/0,10/0,13/0,14/1,15/2,16/0,17/1,18/2}{
	\pgfmathsetmacro{\y}{-.77+\off*.1}
	\draw[blue!80!black, line width=2pt] (\i*.4-.07,\y) --  (\i*.4+.8+.07,\y);
	\draw[blue!80!black, opacity=.5] (y\u) -- (\i*.4+.3,\y+.02);
}
		\foreach \x [count=\i] in {1,1,1,2,1,1,1,2,1,1,1,2,1,1,1,1,3,1,1,1}
		{	\node (x\i) at (\i*.4,-1) {\x};
	}
\end{scope}
\end{tikzpicture}}
    \caption{Top: example of construction of the auxiliary graphs $H^{(k)}$ for $k=2$ and $k=3$. Bottom: conversion of a walk in $H^{(3)}$ into a $4$-realization of $G$.  } \label{fig:illustration_def}
  \end{figure}  
}
We now show that finding a $w$-realization of $G$ is equivalent to finding a connected component of the auxiliary graph that \emph{covers} all edges, as defined below. With this step we remove the need to consider long permutations of vertices, thus reducing the combinatorics to the size of $H$ (that is, $n^{O(w)}$).

\begin{definition}
  An edge $\{x,x'\}$ of $G$ is \emph{covered} by a vertex $y\in V^{(k)}$ if $x=y_i$ and $x'=y_{j}$ for some $1\leq i<j\leq k$; it is covered by an edge $\{y,y'\}$ in $E^{(k)}$ if $y= x_{1:k}$, $y'=x'_{1:k}$ with $x=x_1$, $x'=x'_k$ and $x_{2:k}=x'_{1,{k-1}}$.
  Edge $\{x,x'\}$ is covered by a subgraph of $H^{(k)}$ if it is covered by at least one vertex or one edge in this subgraph.
\end{definition}
\begin{lemma}\label{lem:GU_forward}
  If $G$ has a $w$-realization, then $H^{(w-1)}$ has a connected component covering all edges.
\end{lemma}
\begin{proof}
Write $k=w-1$.
 Let $P=x_{1:\ell}$ be a $w$-realization of $G$. For each $i\in [\ell-k+1]$, let $y_i=x_{i:i+k-1}$. Since $P$ is a $w$-realization, we have $\{x_p,x_q\}\in E$ for each $i\leq p< q\leq i+w-1$, so $y_i$ is a vertex of $V^{(k)}$ for each $i\in [\ell-k+1]$ and $(y_i, y_{i+1})$ is an edge of $E^{(k)}$ for each $i\in [\ell-k]$. 
 Thus, $V'=\{y_1, \ldots, y_{\ell-k+1}\}$ induces a connected subgraph of $H^{(k)}$. 
 Moreover, each edge $\{x,x'\}$ is realized in $P$ with $x=x_i$ and $x'=x_j$, $i+1\leq j\leq i+w-1$.
 We distinguish three cases:
(a)  if $j\leq i+k-1=i+w-2$ and $i\leq \ell-k+1$, then  $\{x,x'\}$ is covered by vertex $y_i \in V'$. (b) if $j=i+w-1$, then  $\{x,x'\}$ is covered by $\{y_i, y_{i+1}\}$ in $H^{(k)}[V']$. (c) if $i>\ell-k+1$, then  $\{x,x'\}$ is covered by vertex $y_{\ell-k+1} \in V'$.

Overall, $\{x,x'\}$ is covered by the connected subgraph $H^{(k)}[V']$, so some connected component of $H^{(k)}$ covers all edges of $G$.
\end{proof}
The following definition gives the main algorithmic tool to build a $w$-realization from any walk in the auxiliary graph.

\begin{definition}
  
  Let  $\{y,y'\}\in E^{(k)}$ with  $y=x_{1:k}$, $y'=x'_{1:k}$ and $x_{2:k}=x'_{1:k-1}$. 
  The \emph{addition $a_{y}(y')$ of $y'$ with respect to $y$} is the single-character string  $x'_{k}$.   
  The \emph{addition $a_{y'}(y)$ of $y$ with respect to $y'$} is the length-$k$ string  $x_{1:k}$.  
  Given a walk $P=(y_1,\ldots, y_\ell)$ of $H^{(k)}$, the \emph{underlying sequence of $P$}, denoted $S_P$ is the concatenation $y_1\cdot a_{y_1}(y_2)\cdot a_{y_2}(y_3)\cdots a_{y_{\ell-1}}(y_\ell)$.
\end{definition}
\begin{lemma}\label{lem:GU_reverse}
  
  For a path $P$, the underlying sequence of $P$ realizes (with window size $k+1$) exactly the edges covered by the subgraph $H^{(k)}[P]$.
\end{lemma}
\begin{proof}
  We first prove the following claim: for any $\{y,y'\}\in E^{(k)}$, the string $y\cdot a_{y}(y')$ ends with $y'$, and realizes exactly the  edges covered by $y$, $y'$ and $\{y,y'\}$ with window $k+1$.
  
  Let $x_{1:k}=y$ and $x'_{1:k}=y'$.
  We distinguish the \emph{forward} case with  $x_{2:k}=x'_{1:{k-1}}$ from the \emph{backward} case with  $x'_{1:{k-1}}=x_{2:{k}}$. 
  
  First note that string $y\cdot a_{y}(y')$ starts with $y$ (by construction) and ends with $y'$ (this is clear in the backward case, and follows from $x_{2:k}=x'_{1:{k-1}}$ in the forward case). Thus, the first size-$k$ window of $y\cdot a_{y}(y')$ realizes edges $\{x_i,x_j\mid 1\leq i<j\leq k\}$ which are exactly edges covered by $y$. Similarly, the last size-$k$ window of $y\cdot a_{y}(y')$ realizes edges $\{x'_i,x'_j\mid 1\leq i<j\leq k\}$ which are exactly edges covered by $y'$. 
  
  We now enumerate the remaining edges realized by size-$(k+1)$ windows of $y\cdot a_{y}(y')$. In the forward case, $\{x_1, x'_{k}\}$ is the only such remaining edge, and it is the one covered by $\{y,y'\}$.
  In the backward case, the remaining edges are those of the form  $\{x_i, x'_{j}\}$ with $1 \leq j<i\leq k$.  
  For $i=k$ and $j=1$,  $\{x_i,x'_j\}$ is the edge covered by $\{y,y'\}$. For $j>1$,  $\{x_i, x'_{j}\}=\{x_i, x_{j-1}\}$ is covered by $y$, and finally for $i<k$,  $\{x_i, x'_{j}\}=\{x'_{i+1}, x'_{j}\}$ is covered by $y'$.
  Overall, the edges realized by all size-$(k+1)$ windows of $y\cdot a_{y}(y')$ are exactly those covered by $y$, $y'$ and $\{y,y'\}$. This completes the proof of the claim.

  Now for the  Lemma statement, let $P=(y_1,\ldots, y_\ell)$ and $S_P$ be the underlying sequence of $P$. Note that $S_P$ contains all  $y_ia_{y_i}(y_{i+1})$ as substrings by construction and each size-$(k+1)$ widows of $S_P$ appear as a size-$(k+1)$ window of some $y_ia_{y_i}(y_{i+1})$. Thus $S_p$ realizes (with window size $k+1$) exactly the edges realized by some  $y_ia_{y_i}(y_{i+1})$, which in turn are exactly the edges covered by vertices $y_i$, $1\leq i\leq \ell$ and edges $\{y_i, y_{i+1}\}$, i.e. edges covered by the subgraph $H^{(k)}[P]$.
\end{proof}

\begin{lemma}\label{lem:GU_XP}
  A graph $G=(V,E)$ is $w$-realizable if and only if $H^{(w-1)}$ has a connected component covering all edges. Moreover, a realization (if any) can be computed in time $O(n^w)$.
\end{lemma}
\begin{proof}
  The forward direction is proven in Lemma~\ref{lem:GU_forward}. The reverse follows from Lemma~\ref{lem:GU_reverse}: pick such a connected component of $H^{(w-1)}$, and let $P$ be a walk covering all edges of this component: the underlying sequence of $P$ is a $w$-realization of $G$. Walk $P$ can be computed with a DFS, by visiting all edges incident to each successive node, which can be performed linearly with respect to the number of edges in $H$. There are at most $n^w$ such edges (with $n=|V|$) since each edge involves at most $w$ distinct vertices of $G$, so we obtain the desired time bound.
\end{proof}

\begin{remark}\label{rem:GU_to_DU}
  For digraphs, the aforementioned procedure can easily be translated using a directed definition of graph $H^{(k)}$, where edge $\{u_{1:k},v_{1:k}\}$ is replaces by arc $(u_{1:k},v_{1:k})$ in Definition~\ref{def:aux_graph_H}. See Figure~\ref{fig:second_procedure}. One then needs to produce a directed walk visiting strongly connected components and covering all edges of $G$. However, there can be exponentially many walks between components in the auxiliary graph, and the choice of which components to visit leads to an exponential blow-up in the complexity of the problem.
\end{remark} 
\begin{figure}
  \captionsetup[subfigure]{justification=centering}
  	\centering
%	 \begin{subfigure}[h]{0.45\textwidth}
\centering
%\begin{minipage}[c]{0.3\linewidth}
%\subcaptionbox{Original graph $G$}[0.3\textwidth]
\subfloat[$G$]
{\begin{tikzpicture}
	\begin{scope}[xshift=4cm]
	\node[main node] (1) {1};
	\node[main node] (2) [right = 1.2cm  of 1] {2};
	\node[main node] (3) [below = 1.2cm  of 2] {3};
	\node[main node] (4) [left = 1.2cm  of 3] {4};
	
	\draw[->] (4) edge node[right] {} (1);
	\draw[<->] (3) edge node[below] {} (4);
	%\draw[<->] (4) edge node[right] {} (3);
	\draw[<->] (2) edge node[right] {} (3);
	%\draw[->] (3) edge node[right] {{\tiny $2$}} (2);
	\draw[->] (4) edge node[right] {} (2);
	\draw[->] (2) edge node[right] {} (4);
	\draw[->] (3) edge node[right] {} (1);
	
	\end{scope}
	\end{tikzpicture}}\qquad \qquad
%\subcaptionbox{Graph $H$ is not a $2$-sequence graph}[0.33\textwidth]
\subfloat[$H^{(2)}$]
{\begin{tikzpicture}
	\begin{scope}[xshift=4cm]
	\foreach \s in {1,...,\n}
	{
		\node[main node]  (\getnthelement{\s}) at ({360/\n * (\s - 1)}:\verySmallRadius) {\scriptsize{\getnthelement{\s}}};
		arc ({360/\n * (\s - 1)+\margin}:{360/\n * (\s)-\margin}:\verySmallRadius);
	}
	\draw[, ->] (42) edge node[right] {} (23);
	\draw[, ->] (23) edge node[right] {} (34);
	\draw[, ->] (23) edge node[right] {} (34);
	\draw[ ->] (34) edge node[right] {} (41);
	\draw[, ->] (34) edge node[right] {} (42);	
	\draw[, ->] (43) -- (31); 
	\draw[, ->] (43) -- (32); 
	\draw[, ->] (32) -- (24); 
	
	\end{scope}
	\end{tikzpicture}}\qquad \qquad
%\subcaptionbox{DAG $R(H^{(1)})$}[0.33\textwidth]
\subfloat[$R(H^{(2)})$]
{\begin{tikzpicture}
	\begin{scope}[xshift=4cm]
	%	\foreach \s in {1,...,\n}
	%	{
	%		\node[reduced node]  (\getnthelement{\s}) at ({360/\n * (\s - 1)}:\verySmallRadius) {\getnthelement{\s}};
	%		arc ({360/\n * (\s - 1)+\margin}:{360/\n * (\s)-\margin}:\verySmallRadius);
	%	}
	
	\newcommand{\s}{1}
	\node[reduced node]  (\getnthelement{\s}) at ({360/\n * (\s- 1)}:\verySmallRadius) {\scriptsize{\getnthelement{\s}}};
	arc ({360/\n * (\s - 1)+\margin}:{360/\n * (\s)-\margin}:\verySmallRadius);
	
	\renewcommand{\s}{2}
	\node[reduced node]  (\getnthelement{\s}) at ({360/\n * (\s - 1)}:\verySmallRadius) {\scriptsize{\getnthelement{\s}}};
	arc ({360/\n * (\s - 1)+\margin}:{360/\n * (\s)-\margin}:\verySmallRadius);
	
	\renewcommand{\s}{4}
	\node[reduced node]  (\getnthelement{\s}) at ({360/\n * (\s - 1)}:\verySmallRadius) {\scriptsize{\getnthelement{\s}}};
	arc ({360/\n * (\s - 1)+\margin}:{360/\n * (\s)-\margin}:\verySmallRadius);
	
	\renewcommand{\s}{6}
	\node[reduced node]  (\getnthelement{\s}) at ({360/\n * ( \s- 1)}:\verySmallRadius) {\scriptsize{\getnthelement{\s}}};
	arc ({360/\n * (\s - 1)+\margin}:{360/\n * (\s)-\margin}:\verySmallRadius);
	
	\renewcommand{\s}{8}
	\node[reduced node]  (\getnthelement{\s}) at ({360/\n * ( \s- 1)}:\verySmallRadius) {\scriptsize{\getnthelement{\s}}};
	arc ({360/\n * (\s - 1)+\margin}:{360/\n * (\s)-\margin}:\verySmallRadius);
	
	\renewcommand{\s}{4}
	\node[reduced node]  (34234) at ({360/\n * ( \s)}:\verySmallRadius) {\scriptsize{34234}};
	arc ({360/\n * (\s - 1)+\margin}:{360/\n * (\s)-\margin}:\verySmallRadius);
	%		31,24,23,43,42,41,34,32,C
	
	\draw[rededge, ->] (34234) edge (41); 
	%\draw[->] (3) edge[bend right=20] node[above] {{\tiny $1$}} (1);
	\draw[rededge, ->] (43) -- (31); 
	\draw[rededge, ->] (43) -- (32); 
	\draw[rededge, ->] (32) -- (24); 
	
	\end{scope}
	\end{tikzpicture}}
%	\subcaptionbox{List of walks in $R(H)$: only one is $3$-admissible.}[0.45\textwidth]
%	{
%		\begin{tikzpicture}
%		\begin{scope}[xshift=4cm]
%		{\footnotesize
%			%\node[reduced node, fill=white, opacity=1]  (1) {1};
%			\node[text width=5cm]  at (0,0)  (1) {34234, 41: $3$-admissible, authentic sequence 342341};
%			\node[text width=5cm, below=0.1 cm of 1]  (2) {43, 31};
%			\node[text width=5cm, below=0.1cm  of 2]  (3) {43, 32};
%			\node[text width=5cm, below=0.1cm  of 3]  (4) {43, 32, 24};
%		}
%		\end{scope}
%		\end{tikzpicture}
%	}
  \caption{Adaptation of the algorithm for \GU-\Realizable{w} to \DU-\Realizable{w}, here with $w=3$ (see Remark~\ref{rem:GU_to_DU}). The walk 34234, 41 in $R(H^{(2)})$ gives a $3$-realization: $3 \, 4 \, 2 \, 3 \, 4 \, 1$. However, finding such a walk takes exponential time in the directed case. } \label{fig:second_procedure}
\end{figure}
\subsubsection{Complexity lower bound for parameter \texorpdfstring{$w$}{w}}
\label{sec:GU_W1_hard}

We now prove Theorem~\ref{theorem:cliquereduction}, restated below,  using a    
polynomial time parameterized reduction from
\Clique{} (decide whether a graph contains $k$ vertices inducing all possible ${k \choose 2}$ edges). See Figure~\ref{fig:reduction_GU} for an illustration.
\theoremcliquereduction*
\begin{proof}
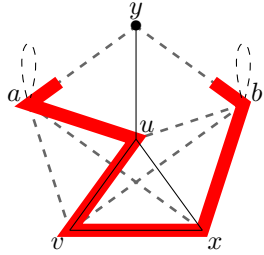
\begin{figure}
    \centering
%\documentclass{standalone}
%\input{common}
%\begin{document}
    \begin{tikzpicture}[inner sep=0,every loop/.style={min distance=10mm,looseness=10}]
        \foreach \l [count=\i] in {a,v,x,b,y}{
            \node[label=\i*72+90:$\l$] (\l) at (\i*72+90:1.5) {$\bullet$};
        }
        \node[label=65:$u$] (u) at (0,0) {$\bullet$};
        \begin{pgfonlayer}{background}
        
        \foreach \q in {u,v,y,x} {
            \draw[, black!60] (a.center)--(\q.center);            
            \draw[, black!60] (b.center)--(\q.center);
        }
        \end{pgfonlayer}
        
        \draw[red, line width=7pt, opacity=.3,line join=round] ($(a)+(36:.5)$) -- (a.center) --(u.center)--(v.center)--(x.center)--(b.center) -- ($(b)+(72*2:.5)$);
        
        \foreach \u/\v in {y/u,u/v,v/x,x/u}{
            \draw[line width=2.5pt] (\u.center) --(\v.center);
        }
        \path (a) edge [loop above,] (a);
        \path (b) edge [loop above,] (b);
        
    \end{tikzpicture}
%\end{document}
    \caption{\label{fig:reduction_GU}Illustration of the reduction for Theorem~\ref{theorem:cliquereduction}, with $w=4$. The source graph $G$ has vertices $\{u,v,x,y\}$ with the bold edges. Vertices $a$ and $b$ and thin edges are added in the reduction. A realization follows a path visiting both vertices $a$ and $b$: the first $w$ vertices of the transition between $a$ and $b$ (highlighted in red) must form a clique in the graph, yielding a $(w-1)$-clique in the original graph.
    }
\end{figure}
Let $G=(V, E)$ be a simple graph. Let $G'$ be a graph constructed from $G$ by adding two nodes $a$ and $b$ with self loops, such that $a$ and $b$ are connected to each vertex of $G$. Let $k$ be a positive integer and $w=k+1$.
We will show that $G$ has a $k$-clique if and only if $G'$ is $w$-realizable.

First, let us suppose that $G$ has a $k$-clique.  Let $C$ be an arbitrary sequence of the vertices of one of its $k$-cliques. Let $v_1, \, \dots, \, v_{|V|}$ be the vertices of $G$  and $ \{u_1, u'_1\}, \, \dots , \, \{u_{|E|}, u'_{|E|}\}$ be its edges. We write $A$ (resp. $B$) for the string containing $w$ successive copies of $a$ (resp. $b$).
Then, the following sequence is a $w$-realization of $G'$:
$$ A \; u_1 \; u'_1 \;  A  \; u_2 \; u'_2 \;  A \,  \dots \,  A \; u_{|E|} \;  u'_{|E|} \;  A \quad C \quad  B \; v_1  \; B \; v_2 \;  B \; \dots B \; v_{|V|} $$

Now let us suppose that $G'$ is $w$-realizable and let  $x=x_1, \, \dots, \, x_p$ be a $w$-realization of $G'$. Without loss of generality, we can suppose $a$ appears before $b$ in $x$. Let $i_b$ be the index of the first appearance of $b$ and let $i_a$ be the largest index of the appearance of $a$ before $i_b$.  Then $i_b -i_a \geq w$, since there is no edge between $a$  and $b$. Furthermore, since $G$ is simple, there cannot be two repetitions of a vertex in the sequence $x_{i_a+1}, \dots, x_{i_a +w -1}$.  Due to the definition of a sequence graph, all vertices $\{x_{i_a+1}, \dots, x_{i_a +w -1}\}$  are connected, forming a clique in $G$ of  size $w-1=k$, which ends the proof.
\end{proof}

\subsection{NP-hardness of directed and weighted variants
for constant window size \texorpdfstring{$w\ge 3$}{w>=3}}\label{sec:nphardnessproofs}

We prove in this section the NP-hardness of the remaining three variants (\DU,  \GW, \DW) for constant window size greater than 3.
\begin{restatable}{proposition}{propositioncomplexityvariants}
\label{proposition_complexity_variants}
{\DU-\Realizable{}$_w$}, {\GW-\Realizable{}$_w$}, and {\DW-\Realizable{}$_w$} are all NP-hard for any $w\geq 3$.
\end{restatable}

We prove each case directly or indirectly by reduction from restricted versions of {\sf Hamiltonian Path}. We first verify the NP-hardness of these variants (see Lemma \ref{lemmaHamiltonianVariant} below). 
We then give the reduction for the unweighted case (see Lemma \ref{lemmaDUrealizable} in Section~\ref{sec:nphardnessUnweighted}), for which we introduce an intermediate variant with \emph{optional} arcs. Finally we give a reduction for both weighted cases in Section~\ref{sec:nphardnessWeighted}, using the same construction for both directed and undirected cases (simply ignoring arc orientations in the latter case,  see Lemma \ref{lemmaXWrealizable}).

We consider slightly constrained versions of {\sf Hamiltonian Path} where we require that the input graph contains up to two degree-one vertices. More formally, we reduce from the following restrictions, which are known to be NP-hard (these are folklore results, included here for completeness):

\begin{lemma} \label{lemmaHamiltonianVariant}
{\sf Hamiltonian Path} is NP-hard even on the graph classes defined by the following restrictions:
\begin{itemize}
\item {\bf HP1:} The input graph has no self-loop, is directed and has a source vertex $s$ (i.e. with in-degree 0)
\item {\bf HP2:} The input graph has no self-loop, is undirected and has two degree-1 vertices $s$ and $t$.
\end{itemize}

\end{lemma}
\begin{proof}
The reduction to HP1 is from {\sf Hamiltonian Cycle} in directed graphs: pick any vertex $v$ and duplicate it into $v_1,v_2$. Each arc $(v,u)$ becomes $(v_1,u)$ and each arc $(u,v)$ becomes $(u,v_2)$. Then $v_1$ is a source vertex and any cycle in the original graph is equivalent to a path  from $v_1$ to $v_2$ in the new graph.

The reduction to HP2 is from {\sf Hamiltonian Cycle} in undirected graphs: pick any vertex $v$ and duplicate it into $v_1,v_2$. Each edge $\{u,v\}$ becomes two edges $\{u,v_1\}$ and $\{u,v_2\}$. Add pending vertices $s$ and $t$ connected to $v_1$ and $v_2$ respectively.  Then any cycle in the original graph is equivalent to a path in the new graph with $\{s,v_1\}$ at one end and $\{t,v_2\}$ at the other.
\end{proof}

\subsubsection{Reduction for \DU-Realizable}
\label{sec:nphardnessUnweighted}
In the directed and unweighted setting, we use the following intermediate generalization which allows \emph{some} arcs to be ignored in the realization. For convenience in the final reduction, we further assume that the first $w-1$ elements of the sequence are given in input. 
\Problem{\textsf{OptionalRealizable$_w$}}%
{problemoptionalrealizable}%
{directed unweighted graph $D=(V,A)$ without self-loops, a subset $A_c\subseteq A$ of \emph{compulsory arcs}, a \emph{starting sequence} $P=(s_1,\dots,s_{w-1})$ of $w-1$ distinguished vertices of $V$.}
{True if there is a sequence $S$, starting with $P$, such that the graph of $S$ with window size $w$ contains only arcs in $A$ and (at least) all arcs in $A_c$; False otherwise.}
%\problemoptionalrealizable*

Note the following similarity between \textsf{OptionalRealizable$_2$} and {\sf Hamiltonian Path}: in the former some arcs are optional and other are compulsory, while in the latter all arcs are optional but vertices are compulsory. This constraint is easily implemented in  \textsf{OptionalRealizable$_2$} by duplicating each vertex and adding a compulsory arc between the two copies. Intuitively, this is the main building block of our reduction  for \textsf{OptionalRealizable$_w$} (Lemma~\ref{lemmaoptionalrealizable} below). With $w\geq 3$, we further introduce $w-2$ \emph{padding vertices} that will take place between any two consecutive vertices of the path, so that the window overlaps successive pairs of vertices along the path. These padding vertices also help enforce that no vertex is visited more than once.

\begin{lemma}\label{lemmaoptionalrealizable}
For any fixed $w\geq 3$, \textsf{OptionalRealizable$_w$} is NP-hard.
\end{lemma}

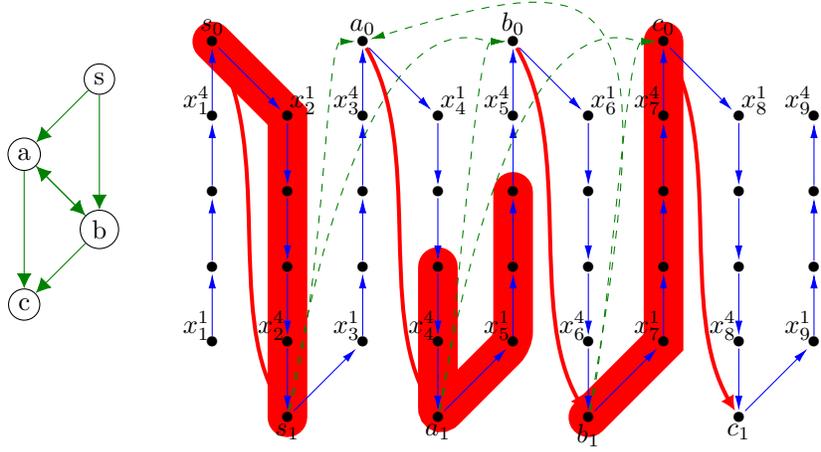
\begin{figure}
    \centering

 %\documentclass{standalone}
%\input{common}
%\begin{document}
\begin{tikzpicture}[base/.style={green!50!black,-{>[width=1mm]}}, orig/.style={dashed, base},inner sep=0, path/.style={blue,-{>[width=1mm]}}]

\begin{scope}[ every path/.style={base},
shift={(-2.5,1.5)}, inner sep=2pt, every node/.style={draw, circle, black}]
\node (s) at (2,3){s};
\node (a) at (1,2){a};
\node (b) at (2,1){b};
\node (c) at (1,0){c};
\draw[->] (s)--(a);
\draw[->] (s)--(b);
\draw[->] (a)--(b);
\draw[->] (b)--(a);
\draw[->] (b)--(c);

\draw[->] (a)--(c);
\end{scope}

\foreach \label [count=\i] in {s,a,b,c}{
    
    \node[label=north:$\label_0$] (v\i0) at (\i*2-1,5) {$\bullet$};        
    \node[label=south:$\label_1$] (v\i1) at (\i*2,0) {$\bullet$};
    
    \draw[red, line width=1.5pt, out=-60, in=120,->, looseness=.8] (v\i0) to(v\i1);
}

\foreach \c in {1,...,9} { 
    \pgfmathparse{mod(\c, 2)} 
    \let\parity\pgfmathresult 
    \foreach \l [remember=\l as \prev] in {1,2,3,4} {
        \node (x\c\l) at (\c,\l) {$\bullet$};
        \ifthenelse{\l>1}{
            \ifthenelse{\equal{\parity}{0.0}}{    
                \draw[path] (x\c\l) -- (x\c\prev);
            }{         
                \draw[path] (x\c\prev) -- (x\c\l);
            }
        }{}
    }  
}
\foreach \c/\i [remember=\c as \prev] in {1/1,3/2,5/3,7/4} {
    \pgfmathtruncatemacro{\nc}{\c +1}
    \pgfmathtruncatemacro{\nnc}{\c +2}    
    \draw[path] (x\c4) -- (v\i0);
    \draw[path] (v\i0) -- (x\nc4);
    \draw[path] (x\nc1) -- (v\i1);
    \draw[path] (v\i1) -- (x\nnc1);
    \node [label=north west:$x_\c^1$] at (x\c1){};    
    \node [label=north west:$x_\c^4$] at (x\c4){};
    \node [label=north east:$x_\nc^1$] at (x\nc4){};    
    \node [label=north west:$x_\nc^4$] at (x\nc1){};
}
    \node [label=north west:$x_9^1$] at (x91){};    
    \node [label=north west:$x_9^4$] at (x94){};
\foreach \s/\t/\a in {1/2/75:2,2/3/75:1,3/4/75:1,2/4/80:6,1/3/80:6} {
    \draw[orig] (v\s1) .. controls +(\a) and +(-180:.5) .. (v\t0);
}\draw[orig] (v31) .. controls +(80:6) and +(17:3.5) .. (v20);
\begin{pgfonlayer}{background}
\begin{scope} [
line width=15pt, 
line cap=round, 
red,
opacity=.2
]
\draw (v11.center) -- (x24.center) -- (v10.center);
\draw[rounded corners] (x42.center)  -- (v21.center)-- (x51.center) -- (x53.center);
\draw (v31.center) -- (x71.center) -- (v40.center);

\end{scope}
\end{pgfonlayer}
\end{tikzpicture}
%\end{document}

    \caption{Left: an instance $G$ of {\sf Hamiltonian Path} with a source vertex $s$ and solution $(s,a,b,c)$. Right: the corresponding instance $G'$ of \textsf{OptionalRealizable$_6$}. Heavy (red) arcs are compulsory, light (blue) arcs are a solution path in the graph, dashed (green) arcs are optional arcs issued from the intput graph. Other optional arcs are not depicted.  Three size-6 windows are overlined: A window using $v_0$ and $v_1$ realizes the compulsory arc for vertex $v$, a window using $u_1$ and $v_0$ enforce that the arc $(u,v)$ exists in $G$, other windows with $w-1$ separator vertices $x_p^i$ help structure the whole sequence. }
    \label{fig:optrealHard}
\end{figure}

\begin{proof}
By reduction from {\sf Hamiltonian Path} (see Lemma~\ref{lemmaHamiltonianVariant}, HP1).
Given a directed graph $G=(V,A)$ with a source vertex $s$ and no self-loop,  build an instance of \textsf{OptionalRealizable$_w$} with directed unweighted graph $G'=(V',A')$, compulsory arcs $A'_c$ and starting sequence $P$ as follows (see Figure~\ref{fig:optrealHard} for an example).

 We introduce vertices denoted $v_0,v_1$ for each vertex $v$ of the original graph, as well as a grid of vertices $x^i_p$ for each $ p\in[2n+1]$, $i\in [w-2]$. The overall vertex set is thus 
$$V'=\{x^i_p \mid  p\in[2n+1], i\in [w-2] \}\ \cup\ \bigcup_{v\in V}\{v_0, v_1\}$$
  
The set of compulsory arcs is $A'_c = \{(v_0,v_1) \mid v\in V\}$. We further introduce the following optional arcs:
    \begin{itemize}
        \item arc $(u_1,v_0)$ for each $(u,v)$ in $A$
        \item arcs $(x^i_{2p-1}, v_0)$, $(v_0, x^i_{2p})$, $(x^i_{2p}, v_1)$, $(v_1, x^i_{2p+1})$ for each $v\in V$, $p\in[n]$, $i\in[w-2]$.
        \item arc $(x^i_p, x^j_p)$ for $i<j$ and $p\in[ 2n+1]$; 
        \item arc $(x^i_p, x^j_{p+1})$ for $j\leq i$ and $ p\in[2n]$  
    \end{itemize}
The starting sequence is defined as $P=(x_1^1, \ldots, x_1^{w-2}, s_0)$. Note that the resulting graph has no self-loop.

\begin{claim} $(G',P, A'_c)$ is a yes-instance for \textsf{OptionalRealizable$_w$} 
$\Leftrightarrow$ $G$ admits an Hamiltonian path
\end{claim}

$\Leftarrow$
Let $v^1, v^2, \ldots, v^n$ be an Hamiltonian path of $G$, so that $v_0^p,v_1^p$ are the corresponding vertices in $G'$, $p\in[n]$. Without loss of generality, since $s$ has degree 1, $v^1=s$. Define the following sequences:
\begin{align*}
X^p &= x_{2p-1}^1 \ldots x_{2p-1}^{w-2} v^p_0 x_{2p}^1 \ldots x_{2p}^{w-2} v^p_1 \quad \text{ for }p\in[n]\\
  X^{n+1}&=x_{2n}^1\ldots x_{2n}^{w-2}\\
  S&=X^1\ldots X^nX^{n+1}
  \end{align*}
We show that $S$ is a solution for \textsf{OptionalRealizable$_w$}($G',P, A_c$). By construction $S$ starts with $P$. 
Further, for each compulsory arc $(v_0,v_1)$, $v$ is part of the Hamiltonian path so there is $p$ such that $v=v^p$:
then compulsory arc $(v_0,v_1)$ is realized in subsequence $X^p$. Finally, it can be checked that the graph of $S$ contains only arcs of $A'$. Indeed, the sequence uses the following arcs: $(v_0,v_1)$ for each $v$ (which are compulsory arcs), arcs $(v^p_1,v^{p+1}_0)$ for each arc $(v^p,v^{p+1})$ of the Hamiltonian path, so  $(v^p,v^{p+1})\in A$ and  $(v^p_1,v^{p+1}_0)\in A'$, arcs with an endpoint $v_i$ and an endpoint $x_p^j$ (which satisfy the parity conditions so they belong to $A'$), and finally arcs of the form $(x_i^p, x_j^q)$, either with $q=p$ (in which case $i<j$) or with $q=p+1$ (in which case by the window size we have $j\leq i$): both kinds are also in $A'$.

$\Rightarrow$
Consider a sequence $S$ solution for  \textsf{OptionalRealizable$_w$}($G',P, A_c$). 
We first show that for each occurrence of $x^i_{p}$ in $S$ (except for $p=2n+1$), the next $w-1$  characters in $S$ are necessarily $x^{i+1}_{p} \ldots x^{w-2}_p v_q x^{1}_{p+1}\ldots x^{i}_{p+1}$ for some $v_q$ with $\in V$ and $q\in\{1,2\}$. 
To this end, consider some size-$w$ window $x S'$ in $S$, where $x=x^i_{p}$ for some $ i\in[w-2]$, $p\in[2n]$ (note that $p\neq 2n+1$). Let $T=x^{i+1}_p\ldots x^{w-2}_p$ and $U=x^{1}_{p+1}\ldots x^{i}_{p+1}$ (note that $T$ is possibly empty). $T$ and $U$ are seen both as strings and as sets of vertices.
The out-neighborhood of $x^i_p$ contains all vertices of $T\cup U$,  as well as all vertices $v_q$ for $v\in V$, where  $q=0$ if $p$ is odd and $q=1$ if $p$ is even. Since there are $k-2$ vertices in $T\cup U$, and no vertex has a self-loop, then by the pigeon-hole principle string $S'$ must contain at least one vertex $v_q$, $v\in V$. Since there are no arc $(a_q, b_q)$ for $a,b\in V$, $S'$ contains exactly one such vertex $v_q$, thus it also contains all vertices of $T \cup U$. Based on the direction of the arcs in $T\cup U\cup\{v_q\}$, it follows that $S' = T\cdot v_q \cdot U$. 

Let $X_p$ be the string $x^1_p \ldots x^{w-2}_p$. From the arguments above, and the fact that $S$ starts with $X_1$ (since $P$ uses vertices of $X_1$), there exist vertices  $u_1,\ldots u_{2n}$ in $\bigcup_{v\in V} \{v_0,v_1\}$ such that 
$$S=X_1 {u_1} X_2 u_2 X_3 u_3 \ldots u_{2n} X_{2n+1}$$
From the window size $w$, there must exist an arc $(u_p,u_{p+1})$ for each $p\in[2n-1]$. So if $u_p=v_0$ for some $p\in [2n-1]$ and $v\in V$, then $u_{p+1}=v_1$ (for the same $v$). Moreover, if $u_p=v_1$ for some $p\in [2n-1]$ and $v\in V$, then $u_{p+1}=v'_0$ for some vertex $v'\in V$ with $(v,v')\in A$. Overall vertices $u_i$ alternate between $\{v_0\mid v\in V\}$ and $\{v_1\mid v\in V\}$, starting with $u_1=s_0$ (by the starting sequence constraint), so there are vertices $(v^1,\ldots, v^n)$ such that $u_{2p-1}=v^p_0$, $u_{2p}=v^p_1$, and $(v^1,\ldots, v^n)$ is a walk of $G$.

Furthermore, arcs $(v_0,v_1)$ are compulsory for each vertex $v\in V$, so $(v^1,\ldots,v^n)$ must visit all $n$ vertices of $G$: Thus, $(v^1,\ldots, v^n)$ is an Hamiltonian path in $G$.
\end{proof}

We can now prove that \textsf{\DU-Realizable$_w$} is NP-hard by reduction from \textsf{OptionalRealizable$_w$}. The main idea of the reduction is to attach a gadget to the graph, that allows to visit each optional arc in some order; this becomes a prescribed prefix of any realization (denoted $Z$ in the proof). Then, the end of a realization must visit all remaining (compulsory) arcs, but can still use any optional arc thanks to the unweighted setting. 
The gadget is heavily constrained to enforce that the prescribed prefix $Z$ is indeed visited ``as is'' (in particular, without realizing any compulsory arc), and that vertices of the gadget are no longer accessible once $Z$ is over.

\begin{lemma}\label{lemmaDUrealizable}
For any fixed $w\geq 3$, {\DU-\Realizable{}$_w$} is NP-hard.
\end{lemma}
\begin{proof}
%Consider the following intermediate problem:
%\todo[inline]{what is the purpose of the definition of "$s, s \in V$"?}
%\textsf{OptionalRealizable$_w$} Given a directed unweighted graph $D=(V,A)$, a subset $A'\subseteq A$ of \emph{compulsory arcs}, two distinguished vertices $s,s'\in V$. Is there a sequence $S$ such that the graph of $S$ contains only arcs in $A$ and (at least) all arcs in $A'$.
%%We first prove that this problem is NP-hard, then show how it reduces to \DU-Realizable.
%To do so, we use \textsf{OptionalRealizable$_w$} defined in Problem~\ref{problemoptionalrealizable}. First, for any $w\geq 3$, \textsf{OptionalRealizable$_w$} is NP-hard thanks to Lemma \ref{lemmaoptionalrealizable}. \newline
% ~\\
%Reduction of \textsf{\DU-Realizable$_w$}  from \textsf{OptionalRealizable$_w$} \hfill \break
 Assume that we are given a directed unweighted graph $G=(V,A)$, a subset $A_c\subseteq A$ of compulsory arcs (let $A_o=A\setminus A_c$ be the set of optional arcs), and a starting sequence $P=(s_1\ldots s_{w-1})$ of vertices of $V$. 
The following reduction is illustrated in Figure~\ref{fig:redOption_DU}.

Let $m=|A_o|$, and write $A_o=\{(u_1,v_1), \ldots, (u_m,v_m)\}$. Create $G'$ by adding $w(m+1)+m$ \emph{separator} vertices: $w(m+1)$ vertices $y_p^i$ with $1\leq p\leq m+1$ and $1\leq i\leq w$, and $m$ vertices $z_p$ for $1\leq p\leq m$. Build the following strings (using the product operator for concatenation)
\begin{align*}
  Z&=\left(\prod_{p=1}^m (y_p^1 \ldots y_p^w u_p z_p v_p) \right)y_{m+1}^1 \ldots y_{m+1}^w
Z'& = Z s_1 \ldots s_{w-1}
\end{align*}

The arc set can be concisely defined as follows : take the set $A$ and insert all arcs realized by $Z'$ involving at least one separator vertex.  In details, the additional arcs in $G'$ are the following (where indices $i,j,p$ necessarily satisfy $ i\in[w]$, $j\in[w]$ and $ p\in[m]$):
\begin{itemize}
    \item $(y_p^i, y_p^j)$ for $i < j$  and $(y_p^i, y_{p+1}^j)$ for $j\leq  i-4$, 
    \item $(y_{m+1}^i, y_{m+1}^j)$ for $i < j$  and $(y_{m+1}^i, s_j)$ for $j< i$,
    \item $(y_p^i, u_p)$ for $2\leq i$ and $(u_p, y_{p+1}^j)$ for $j\leq w-3$,   
    \item $(y_p^i, z_p)$ for $3\leq i$ and $(z_p, y_{p+1}^j)$ for $j \leq w-2$,
    \item $(y_p^i, v_p)$ for $4\leq i$ and $(v_p, y_{p+1}^j)$ for $j\leq w-1$,
    \item $(u_p, z_p)$ and $(z_p,v_p)$.      
\end{itemize}
%}

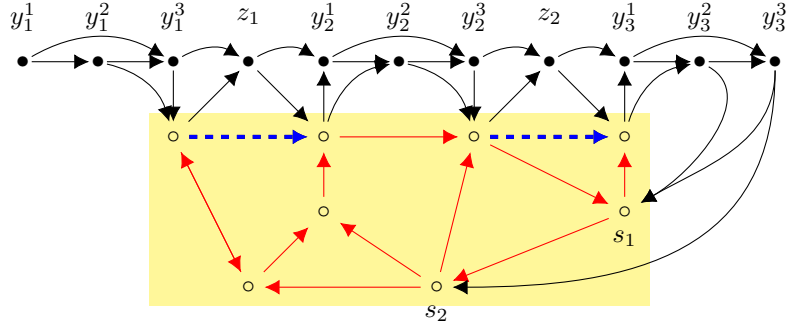
\begin{figure}
    \centering
 %\documentclass{standalone}
%\input{common}
%\begin{document}
    
    \begin{tikzpicture}
    \foreach \x/\y [count=\i] in {1/3,3/3,5/3,7/3,2/1,3/2,4.5/1,7/2} {
        \node (v\i) at (\x,\y) {$\circ$};
    }
    \node[label=south:$s_1$] at (v8) {};
    
    \node[label=south:$s_2$] (s2) at (v7) {};
    
    \foreach \s/\t [count=\i] in {1/2,3/4} {
        \draw[dashed,->, line width=1.5pt] (v\s) --(v\t);
        \draw[dashed,blue, ->, line width=1.5pt] (v\s) --(v\t);
    }
    
    \foreach \s/\t [count=\i] in {1/5,5/1,5/6,6/2,7/6,7/3,2/3,7/5,3/8,8/7,8/4} {
        \draw[red,->] (v\s) --(v\t);
    }
    \foreach \i in {1,...,11} {
           \node[outer sep=0] (t\i) at (\i*1-2, 4) {}; 
           \node at (t\i) {$\bullet$};      
    }
    
    \foreach \v [count=\index, remember=\v as \u (initially t1)] in {t2,t3,v1,t4,v2,t5,t6,t7,v3,t8,v4,t9,t10,t11}{
    \draw[->] (\u)--(\v);
    
  \ifnum \index > 1
    \ifnum \index = 5 \else 
    \ifnum \index = 11 \else 
    \draw[opacity=.3, ->] (\lastu) to[bend left] (\v);
\fi    \fi
  \fi
  \xdef\lastu{\u} 
    }
    \foreach \p in {1,2,3}{
        \foreach \x in {1,2,3} {
            \node at (\p*4+\x-6, 4.6) {$y_\p^\x$};
        }
    }
    
    \foreach \p in {1,2}{
            \node at (\p*4+-2, 4.6) {$z_\p$};        
    }
    
    \draw [out=-90,in=30,->] (t11)to (v8);  
    \draw [out=-30,in=30,->,opacity=.3] (t10)to (v8);  
    \draw [out=-90,in=0,->,opacity=.3,looseness=1.2] (t11)to (v7);   
    \begin{pgfonlayer}{background}
     \node[fit=(v1)(v4)(s2), fill=yellow!50]{} ;
    
    \end{pgfonlayer}
    \end{tikzpicture}
%\end{document}
    \caption{ Reduction from \textsf{OptionalRealizable$_w$} to \textsf{\DU-Realizable$_w$} with $w=3$. The input instance is the highlighted graph with white vertices (including two optional dashed blue  arcs), as well as the starting sequence $(s_1,s_2)$. The reduction adds the black vertices $y_p^i$ and $z_p$, and the corresponding black and grey arcs. Any solution is thus forced to first realize all optional arcs, and then realize the rest of the graph starting with $s_1,s_2$, and including all compulsory arcs and, as needed, some of the optional arcs again.}
    \label{fig:redOption_DU}
\end{figure}

\begin{claim} $G$ has a $w$-realization with optional arcs $\Leftrightarrow$ $G'$ has a $w$-realization
\end{claim}
$\Rightarrow$ Build a realization for $G'$ by concatenating $Z$ with the realization for $G$ starting with $s_1 \ldots s_{w-1}$. All optional arcs of $G'$ are realized in $Z$, all compulsory arcs of $G'$ are realized in the suffix (the realization of $G'$), and all arcs involving a separator are realized in $Z'$. No forbidden arc is realized.

$\Leftarrow$ 
Let $S$ be a realization of $G'$. We prove by induction on $q$, for $1\leq q \leq |Z|$, that 
(i) $S$ and $Z'$ have the same prefix of length-$(q+w-1)$ and (ii) any separator in $Z[1,\ldots, q]$ may only appear in $S[1,\ldots, q]$. 

For $q=1$, this is obtained by the fact that $Z[1]=y_1^1$ has in-degree 0 in $G'$ (so $S$ starts with $y_1^1$ and $y_1^1$ does not appear again in $S$) and its out-neighborhood forms a size-$(w-1)$ tournament corresponding to $Z[2... w]$, so the length-$w$ prefix of $S$ is $Z[1...w]$.  
Consider now $1<q\leq |Z|$. By induction $S$ and $Z'$ have the same prefix of length-$(q+w-2)$, and separators up to position $q-1$ in $Z$ do not have any other occurrence in $S$.  Let $q'=q$ if $S[q]$ is a separator (case A), and $q'=q+1$ otherwise (case B). In both cases, $S[q']$ is a separator, its in-neighborhood contains at least one separator $Z[q-1]$ or $Z[q-2]$, so in particular vertex $S[q']$ may not have any other occurrence in the sequence (otherwise $Z[q-1]$ and/or $Z[q-2]$ would also have two occurrences). Furthermore, the out-neighborhood of $S[q']$ is $N=\{Z'[q'+1],\ldots,Z'[q'+w-1]\}$ without self-loops, so $S[q'+1,\ldots,q'+w-1]$ is a permutation of $N$. In case A, $w-2$ vertices of $N$ are already accounted for (by induction) in $S[q'+1,\ldots,q'+w-2]$, so the remaining vertex $Z'[q'+w-1]$ must be in position $q'+w-1$ in $S$. In case B, elements of $N$ are all in $Z$, so they form a tournament and, again, the next $w-1$ positions in $S$ and $Z'$ must be equal.

Overall, we have $S=Z S'$ with the following properties: the length-$(w-1)$ prefix of $S'$ is the starting sequence $P$, and no separator appears in $S'$. Thus $S'$ realizes only arcs from $G$. Moreover no compulsory arc of $G$ is realized in $Z$, nor with one vertex in $Z$ and one in $S'$ (since such arcs start with a separator), so all compulsory arcs are realized in $S'$. Overall, $G$ is a yes-instance of {\sf OptionalRealizable}$_w$ with sequence $S'$.

\end{proof}

\subsubsection{Reduction for weighted problems}
\label{sec:nphardnessWeighted}

Now, let us prove that {\GW-\Realizable{}$_w$} and {\DW-\Realizable{}$_w$} are NP-hard for all $w\geq 3$, by reduction from the undirected variant of HP2 of \textsf{Hamiltonian Path} (cf Lemma~\ref{lemmaHamiltonianVariant}). We focus on the directed case first, \DW-\Realizable{}$_w$, the undirected case will simply use the underlying graph introduced in this reduction.

\begin{lemma}\label{lemmaXWrealizable}
For any fixed $w\geq 3$, {\DW{}-\Realizable{}$_w$} and {\GW{}-\Realizable{}$_w$} are NP-hard.
\end{lemma}

\paragraph{Reduction for \DW-\Realizable{}.}

\begin{figure}[]
  \centering
  \scalebox{.85}{
 %\documentclass{standalone}
%\input{common}
%\begin{document}
		\begin{tikzpicture}[bul/.style={main node}]
        
    \begin{scope}[shift={(3,5)}]
    
    \node[align=center] at (-.5,3) {Start Gadget:};
    
  %   arc $(s_0',s_0)$ with weight $k$   
  %   loop $(s_0,s_0)$ with weight ${k \choose 2}$    
  %   arc $(s_0,a)$ with weight $k$     
  %   arc $(s_0,s)$ with weight ${k+1 \choose 2}$ 
  %   loops $(s,s)$ and $(t,t)$ with additional weight ${k \choose 2}$    
  %   arcs $(b,t)$ and $(s,b)$ with weight $({k+1\choose 2} + 2k)$    
    
    %\node [bul, label=$u$] at (-2,0) {};
    \node [main node]  (s0) at (-2,0) {$s_0$};
    \node [main node] (s0p) at (-3,1.5) {$s_0'$};
    \node [main node] (a) at (0.5,.5) {$a$};
    \node [main node]  (s) at (1,-0.5) {$s$};
    
    \draw[->] (s0p) edge[bend left=5] node[midway, left]{$k$} (s0);    
    \draw[->] (s0) edge[loop below, min distance=10mm,in=30,out=60] node[midway, right]{${k \choose 2}$} (s0);
    \draw[->] (s0) edge[bend right=5] node[midway, below]{$k$} (a);
    \draw[->] (s0p) edge[bend left=10] node[midway, above]{$1$} (a);    
    \draw[->] (s0) edge[bend right=5] node[midway, below]{${k+1 \choose 2}$} (s);
    
 %   \draw[->] (s) edge[loop below, min distance=10mm,in=30,out=60] node[midway, right]{${k \choose 2}$} (s);

    \end{scope}
    
    \begin{scope}[shift={(6,5)}]
    
    \node[align=center] at (3,3) {Queue Gadget:};

    \node [main node] (a) at (1,0.5) {$a$};
    \node [main node]  (t) at (1,-.5) {$t$};
    %\node [main node]  (t) at (5,0) {$t$};
    \node [main node] (b) at (3.5,0) {$b$};
    
    %\draw[->] (t) edge[loop below, min distance=10mm,in=30,out=60] node[midway, right]{${k \choose 2}$} (t);
    
    %\draw[->] (s) edge[bend right=5] node[midway, below]{${k+1\choose 2} + 2k$} (b);
    %\draw[->] (b) edge[bend left=10] node[midway, above]{${k+1\choose 2} + 2k$} (t);
    \draw[->] (t) edge[bend right=10] node[midway, below]{${k+1 \choose 2}$} (b);
    \draw[->] (a) edge[bend left=5] node[midway, above]{$k+1$} (b);    
    \draw[->] (b) edge[loop below, min distance=12mm,in=100,out=80] node[midway, above]{$(m-n+1){w \choose 2}+2(m-n)$} (b);
    \end{scope}

    \begin{scope}[shift={(-1,0)}]
    \node[align=center] at (3.5,3) {Vertex Gadget\\ (for each vertex $u$, including $s$ and $t$):  };
    %\node [bul, label=$u$] at (-2,0) {};
    \node [main node] (a) at (0,0) {$a$ };
\node [main node] (b) at (6.8,0) {$b$ };
\node [main node]  (u) at (2,0) {$u$ };
\node [main node] (up) at (1,2) {$u'$};
    \draw[->] (u) edge[loop below, min distance=18mm,in=40,out=60] node[midway, above]{$(\delta_u +2){k \choose 2}+{k+1 \choose 2}$} (u);
    \draw[<->] (u) edge[bend left=5] node[midway, left]{$k$} (up);
    \draw[<->] (a) edge[bend right=5] node[midway, above]{$k$} (u);
    %\draw[->] (u) edge[bend right=5] node[midway, above]{$k$} (a);
    \draw[<->] (a) edge[bend left=5] node[midway, below]{$1$} (up);
    %\draw[->] (up) edge[bend right=5] node[midway, above]{$1$} (a);
    
    \draw[->] (b) edge[bend left=5] node[midway, below]{$\delta^{out}_u({w\choose 2} -1)+\delta^{in}_u k$} (u);
    \draw[->] (u) edge[bend left=5] node[midway, above]{$\delta^{in}_u({w\choose 2} -1)+\delta^{out}_u k$} (b);
    %\draw[->] (u) edge[bend right=5] node[midway, above]{$(d_u-1)({k+1\choose 2} + 2k)$} (b);
    
    \end{scope}
    \begin{scope}[shift={(8,0)}]
    \node[align=center] at (1,3) {Edge Gadget\\ (for each $\{u,v\}$):    };
    %\node [bul, label=$u$] at (-2,0) {};
    \node [main node] (u) at (0,1) {$u$};
    \node [main node] (v) at (2,1) {$v$};
    
    \draw[<->] (u) edge[] node[midway, below]{${k+1 \choose 2}$} (v);
    
    \end{scope}

    \end{tikzpicture}
 % \end{document}
    
  }
  \caption{\label{fig:gadgetsDW}Subgraphs used in the reduction from {\sf Hamiltonian Path} to {\sf DW-Realizable}$_3$. Weights on double arcs apply to both directions. Note that arcs $(t,b)$ appear in two different gadgets, so their weights should be summed.
  } 
  %\documentclass{standalone}
%\input{common}
%\begin{document}

\begin{tikzpicture}
    %%% LEFT: input graph
    \node at (-.8,2) {$G$};
    \foreach \l [count=\i] in {s,u,v,w,t} {
    \pgfmathtruncatemacro{\a}{180+60/2-\i*60}
        \node[main node] (u\i) at (\a:1) {$\l$};
        \node[main node] (u\l) at (u\i) {$\l$};
    }
    \foreach \u/\v in {s/u,u/v,v/w,w/t,u/w}{
       \draw  (u\u) --(u\v);
    }
    %%% MIDDLE: first part of the realization
    \begin{scope}[shift={(4,0)}]
    \node at (-1.6,2) {$G'$};
    \foreach \l [count=\i] in {s,u,v,w,t,s_0} {
    \pgfmathtruncatemacro{\a}{180+60/2-\i*60}
        \node[main node] (u\i) at (\a:1) {$\l$};
        \node[main node] (u\l) at (u\i) {$\l$};
        \node[main node] (up\i) at (\a:1.8) {$\l'$};
        \node[main node] (up\l) at (up\i) {$\l'$};
    }
    \foreach \u/\v in {u/w}{   \draw[black!40]  (u\u) --(u\v);}
    \foreach \u/\v in {s/u,u/v,v/w,w/t}{\draw[black, line width=2pt,->]  (u\u) --(u\v);}
    \node[main node] (a) at (0,0) {$a$};
    \foreach \i in {1,...,6}{
    \pgfmathtruncatemacro{\a}{180+60/2-\i*60}
      \ifnum \i<6
        \draw[red]         
            (\a+30:.2) 
            -- (\a+10:1) 
              edge[loop below, min distance=4mm,in=\a+70,out=\a+110] (\a+10:1) 
            -- (\a:1.5);
            \fi
            \draw[red]  (\a:1.5)
            -- (\a-10:1)edge[loop below, min distance=4mm,in=\a-70,out=\a-110] (\a-10:1)   
            --    (\a-30:.2) ;
    }
    \end{scope}
    %%% RIGHT: second part of the realization
    \begin{scope}[shift={(9,0)}]
    \node at (-1.6,2) {$G'$};
    \foreach \l [count=\i] in {s,u,v,w,t,s_0} {
    \pgfmathtruncatemacro{\a}{180+60/2-\i*60}
        \node[main node] (u\i) at (\a:1) {$\l$};
        \node[main node] (u\l) at (u\i) {$\l$};
        \node[main node] (up\i) at (\a:1.8) {$\l'$};
        \node[main node] (up\l) at (up\i) {$\l'$};
    }
    \foreach \u/\v in {u/w}{ \draw[black, line width=2pt,->]  (u\u) --(u\v);}
    \foreach \u/\v in {s/u,u/v,v/w,w/t}{\draw[black!40]  (u\u) --(u\v);}
    \node[main node] (a) at (0,0) {$b$};
    \pgfmathtruncatemacro{\i}{3}
    \pgfmathtruncatemacro{\a}{180+60/2-\i*60}
        \draw[red]     
            (\a+70:.2) 
              edge[loop below, min distance=4mm,in=\a+170,out=\a+130] (\a+70:.2)         
            -- (\a+60:.8)
            edge[loop below, min distance=4mm,in=\a+170,out=\a+130] 
            (\a+60:.8)        
           -- (\a:.2)
            -- (\a-60:.8)
            edge[loop below, min distance=4mm,in=\a-170,out=\a-130] 
            (\a-60:.8)        
           --(\a-70:.2) 
              edge[loop below, min distance=4mm,in=\a-170,out=\a-130] (\a-70:.2)        ;
    \end{scope}
    \end{tikzpicture}
%\end{document}
  \caption{Reduction from {\sf Hamiltonian Path} to \textsf{DW-Realizable}.  
    Left: the input graph $G$ with degree-one vertices $s$ and $t$ (and Hamiltonian Path $(s,u,v,w,t)$). Each edge becomes two arcs (in each direction) in the edge gadgets of the resulting graph $G'$. Center: the first part of the path realizing most edges of $G'$ (in particular those involving vertex $a$ and primed versions of vertices), following the Hamiltonian path. In particular, bold arcs from the input graph are realized. Right: the remaining arcs (such as $(u,w)$, but also $(w,u)$, $(u,s)$, $(v,u)$, etc.) are realized using a succession of round-trips with vertex $b$. Self-loops represent iterations of $k=w-2$ or $w$ occurrences.}
  \label{fig:DWRealizable}
\end{figure}

Given $G=(V,E)$ undirected with degree-1 vertices $s$ and $t$, write $d_u$ for the degree of each vertex $u\in V$, and $k=w-2$ (note that $k\geq 1$ since we chose $w\geq 3$). We write $\delta^{in}_s=d_s$ and $\delta^{in}_u=d_u-1$ for $u\in V\setminus \{s\}$; and  $\delta^{out}_t=d_t$ and $\delta^{out}_u=d_u-1$ for $u\in V\setminus \{t\}$ ($\delta^{in}_u$ and $\delta^{out}_u$ can be seen as the remaining in- and out-degree in the oriented graph where edges are replaced by double arcs after removing an Hamiltonian $s-t$ path). We write $\delta_u=\delta^{in}_u+\delta^{out}_u$.
Build a directed weighted graph $G'=(V',A)$ as follows.
For each $u\in V$, add $u$ and a new vertex denoted $u'$ to $V'$. Create additional dummy vertices $s_0$, $s_0'$, $a$ and $b$. The overall vertex set is thus $V':=\{a,b,s_0,s_0'\} \cup \bigcup_{u\in V}\{u,u'\}$. 
The arcs of $A$ are given in Figure~\ref{fig:gadgetsDW}, as the union of the start gadget, the queue gadget, and the vertex and edge gadgets respectively for each vertex and edge of $G$. An example realization is given in Figure~\ref{fig:DWRealizable}.

\paragraph{Reduction for \GW-\Realizable.}
Build the directed graph $G'$ as above, and let $G_u'$ be the undirected version of $G'$: remove arc orientations, for $u\neq v$ the weight of $\{u,v\}$ is the sum of the weight of $(u,v)$ and $(v,u)$ in $G'$ (the weight of loops is unchanged).

\paragraph{Correctness of the reduction.}
We prove the following three claims:
\begin{claim}\label{claim:i}
$G$ Hamiltonian $\Rightarrow$ $G'$ has a realization
\end{claim}
\begin{claim}\label{claim:ii}
$G'$ has a realization $\Rightarrow$ $G'_u$ has a realization
\end{claim}
\begin{claim}\label{claim:iii}
$G'_u$ has a realization $\Rightarrow$ $G$ is Hamiltonian
\end{claim}

\begin{proof}[Proof of Claim~\ref{claim:i}]
Assume that $G$ has an Hamiltonian path and denote its vertices as $u_1,u_2,\ldots u_n$ according to their position along the path (wlog.,  $u_1=s$ and $u_n=t$).  Let $(v_1,w_1),\ldots (v_{m'},w_{m'})$ be the pairs of connected vertices in $G$ that are not consecutive vertices of the Hamiltonian path (formally, it corresponds to the set $\bigcup_{\{u,v\}\in E}\{(u,v),(v,u)\}\setminus\{(u_i,u_{i+1})\mid 1\leq i<n\}$). Note that there are $m'=2m-(n-1)$ such pairs. We now show that the sequence $S$ defined as follows is a realization of $G$ (recall that $k=w-2$).

\begin{align*} 
S &:= S_{init} S_{path} S_{queue} \quad \text{with} \\
S_{init}&:=
  s_0'\,s_0^k \\
 S_{path}&:=a\,  s^k\,s'\,s^k\, a\, u_2^k\,u_2'\,u_2^k\, a \ldots  a\, u_{n-1}^k\,u_{n-1}'\,u_{n-1}^k\,  a\, t^k\,t'\,t^k\, a\\
 S_{queue}&:=
  b^w\,v_1^k\,b\,w_1^k\,b^w\,v_2^k\, b\, w_2^k \ldots b^w\,v_{m-n}^k\, b\, w_{m-n}^k\, b^w
  \end{align*}

We verify for each gadget that all arcs are indeed realized with the correct weight\footnote{For most arcs, the weights are not actually relevant and are only computed since they must be part of the input. Only weights of arcs incident to $s_0'$, $a$ and vertices $u'$ are used in the rest of the proof (proof of Claim~\ref{claim:iii}). The central point here is that the sequence graph of $S$ is the same for \emph{any} sequence $S$ obtained from an Hamiltonian path of $G$.}, Figure~\ref{fig:graph_fragments} helps compute the weights of short string fragments.
\begin{figure}
\begin{center}
%\documentclass{standalone}
%\input{common}
%\begin{document}
		\begin{tikzpicture}[bul/.style={main node}]

    \begin{scope}[shift={(0,2)}]
    
    \node[align=center] at (-4,.5) {$x^k a y^k \rightarrow$};

    \node [main node] (x) at (-1.2,0.5) {$x$};
    \node [main node]  (y) at (1.2,0.5) {$y$};
    \node [main node] (a) at (0,0) {$a$};
    
    \draw[->] (x) edge[bend right=15] node[midway, below]{$k$} (a);    
    \draw[->] (a) edge[bend right=15] node[midway, below]{$k$} (y);    
    \draw[->] (x) edge[bend left=5] node[midway, above]{${k+1\choose 2}$} (y);    
    \draw[->] (x) edge[loop, min distance=7mm,in=170,out=190] node[midway, below]{${k \choose 2}$} (x);
    \draw[->] (y) edge[loop, min distance=7mm,in=10,out=-10] node[midway, below]{${k \choose 2}$} (y);
    \end{scope}
\begin{scope}[shift={(0,-2)}]

\node[align=center] at (-4,.5) {$b x^k b^w \rightarrow$};

\node [main node] (x) at (-1.2,0.5) {$x$};
\node [main node]  (b) at (1.2,0.5) {$b$};

\draw[->] (x) edge[bend left=15] node[midway, above]{${w\choose 2}-1$} (b);    
\draw[->] (b) edge[bend left=15] node[midway, below]{$k$} (x);    
\draw[->] (x) edge[loop, min distance=7mm,in=170,out=190] node[midway, below]{${k \choose 2}$} (x);
\draw[->] (b) edge[loop, min distance=7mm,in=10,out=-10] node[midway, below]{${w \choose 2}+1$} (b);
\end{scope}
    
    \begin{scope}[shift={(0,0)}]
    	
    	\node[align=center] at (-4,.5) {$b^w x^k b \rightarrow$};

    	\node [main node] (x) at (-1.2,0.5) {$x$};
    	\node [main node]  (b) at (1.2,0.5) {$b$};

    	\draw[->] (x) edge[bend left=15] node[midway, above]{$k$} (b);    
    	\draw[->] (b) edge[bend left=15] node[midway, below]{${w\choose 2}-1$} (x);    
    	\draw[->] (x) edge[loop, min distance=7mm,in=170,out=190] node[midway, below]{${k \choose 2}$} (x);
    	\draw[->] (b) edge[loop, min distance=7mm,in=10,out=-10] node[midway, below]{${w \choose 2}+1$} (b);
    \end{scope}

    \end{tikzpicture}
%  \end{document}
\end{center}
\caption{\label{fig:graph_fragments} Sequence graphs for strings of the form $x^kay^k$, $b^wx^k b$ or $bx^k b^w$ for any $x,y, a,b$, with window size $w=k+2$. }  
\end{figure}
The start gadget corresponds exactly to arcs in $S_{init}$ or overlapping $S_{init}$ and $S_{path}$. 
Regarding the vertex gadget for $u\in V$, $S_{path}$ realizes all arcs involving two distinct vertices among $a,u,u'$. 
$S_{path}$ also yields ${k\choose 2}+{k+1\choose 2}$ self-loops for $u$, and $S_{queue}$ yields the remaining $\delta_u{k\choose 2}$ self-loops (since each vertex appears $\delta_u$ times there). $S_{queue}$ also realizes all arcs between $u$ and $b$. 
For an edge gadget $\{u,v\}$ if $uv$ (resp. $vu$) is part of the Hamiltonian path, then the arc $(u,v)$ (resp. $(v,u)$) is realized in $S_{path}$, otherwise it is realized in $S_{queue}$. Finally, the arcs in the queue gadget are realized either in $S_{queue}$, either as overlapping arcs between $S_{path}$ and $S_{queue}$.
\end{proof}

\begin{proof}[Proof of Claim~\ref{claim:ii}]
Clearly, any realization for $G'$ is a realization for $G'_u$.
\end{proof}

\begin{proof}[Proof of Claim~\ref{claim:iii}]
Pick a realization $S$ of $G'_u$.  Define the weight of a vertex in $G_u$ as the sum of the weights of its incident edges (counting loops twice). From the construction, we obtain the following weights for a selection of vertices:
\begin{itemize}
    \item $s_0'$ has weight $k+1=w-1$
    \item $u'$ has weight $2(w-1)$ for $u\in V$
    \item $a$ has weight $2(n+1)(w-1)$    
\end{itemize}

From the weight of $s_0'$, it follows that this vertex must be an endpoint of $S$ (wlog, $S$ starts with $s_0'$). Then that for any other vertex $v$ with weight $2i(w-1)$, $v$ must have exactly $i$ occurrences in $S$ (in general it can be either $i$ or $i+1$, but if $v$ has $i+1$ occurrences it must be both the first and last character of $S$, i.e. $v=s_0'$: a contradiction). Thus each $u'$ occurs once and $a$ occurs $n+1$ times in $S$.

Each $u'$ occurs once, so order vertices of $V$ according to their occurrence in $S$ (i.e. $V=\{u_1, \ldots, u_n\}$ with $u'_1$ appearing before $u'_2$, etc.). For each $i$, the neighborhood of $u'_i$ in $S$ contains $a$ twice, one $a$ on each side (since there is no $(a,a)$ loop). Other neighbors of $u_i'$ may only be occurrences of $u_i$, so each $u_i'$ belongs to a factor, denoted $X_i$, of the form $a u_i^* u_i' u_i^* a$.
Two consecutive factors $X_{i}, X_{i+1}$ may overlap by at most one character ($a$), and if they do, then there exists an arc $(u_i,u_{i+1})$ in $A$ , hence an edge $\{u_i,u_{i+1}\}$ (since $w\geq 3$) in $E$. There are $n$ such factors $X_{u_i}$, and only $n+1$ occurrences of $a$, so all $a$s except extreme ones belong to the overlap of two consecutive $X_i$s, and there exists an edge $\{u_i,u_{i+1}\}$ for each $i$. Thus $(u_1,\ldots, u_n)$ is an Hamiltonian path of $G$.
\end{proof}

All together, claims~\ref{claim:i}, ~\ref{claim:ii} and ~\ref{claim:iii} show the correctness of the reductions for both \GW-\Realizable and \DW-\Realizable from Hamiltonian Path (HP2) 
since they yield :\begin{eqnarray*}
  G \text{ is Hamiltonian }\Leftrightarrow G' \text{ has a realization}\\
  G \text{ is Hamiltonian }\Leftrightarrow G_u' \text{ has a realization}\\
\end{eqnarray*}

 This completes the proof of Lemma~\ref{lemmaXWrealizable}.

\subsection{Exponential algorithms for weighted problems}\label{sec:effective_general_algo}
In spite of its NP-hardness, the weighted version of \Realizable{} can be solved exactly for moderate instance sizes. In this section, we provide two complementary exponential-time algorithms, respectively based on Integer Linear Programming and Dynamic Programming, to effectively solve our problem in the context of tame instances. The latter runs in $\mathcal{O}(n^w\,2^{wp})$, and produces the total number of realizations.

\subsubsection{Linear integer programming formulation for \DW- and \GW-\Realizable{w}}
Let $G=(V,E)$ be a graph with integer weights $\pi_{e\in E}$. We consider first the directed case \DW, and show how our results extend to \GW at the end of this section. 
In our model, we represent a sequence $x$ over a size-$n$ alphabet $V$è, as a boolean matrix $X \in \mathbb{M}_{n,p}(\{0,1\})$ encoding the sequence $x$:
\begin{equation*}
 X_{v,j}=\begin{cases}
    1  \quad \mathrm{if} \; x_j = v  \\
    0  \quad \mathrm{otherwise}
  \end{cases}
\end{equation*}
The length-$p$ sequences over $V$ are thus in bijection with the boolean matrices such that $ \forall j \in [p], \sum_{v\in V} X_{v,j} = 1 $.

In order to encode sliding window constraints, we define the set $\mathcal C$ of all pairs of positions occurring together in a size-$w$ window, $\mathcal C=\{(i,j) \mid i,j\in [p], i<j<i+w\}$.
We use an intermediary slack variable $y_{i,j}^{e} \in \{0,1\}$ for each $(i,j)\in\mathcal C$ and $e=(v_1,v_2)\in E$ to model the appearance of $v_1, v_2$ together at indices $i,j$ in a size-$w$ window, i.e. $y_{i,j}^{e}$ is equal to $1$ when $v_1$ is located at position $j$ and  $v_2$ at position $j+i$, and 0 otherwise. This is achieved with the following linear constraints
\begin{equation*}
  \begin{alignedat}{2}
    -  X_{v_1, i}  &   && + y_{i,j}^{e}   \leq  0 \\
    &   - X_{v_2, j}  && + y_{i,j}^{e}    \leq  0 \\
    X_{v_1, i} & + X_{v_2, j}  &&  - y_{i,j}^{e}   \leq  1 
  \end{alignedat}
\end{equation*}
For a length-$p$ sequence, the number of possible position pairs $(i,j)\in \mathcal C$ is given by:
$$|\mathcal C| = \sum_{d=1}^{w-1}(p-d) = p(w-1) - \frac{w(w-1)}{2} = (w-1)(p - \frac{w}{2})$$

We also need to forbid missing pairs not forming edges in the graph, which give the following constraint for each $(v_1,v_2)\notin E$ and every $(i,j)\in\mathcal C$
\begin{equation*}\label{constraint_absence}
  X_{v_1,i}  + X_{v_2, j} \leq  1
\end{equation*}
It is worth noting that the value of $p$ can be directly computed from the input, since the weight matrix and window size entirely constraints the realization size. More precisely, $p$ is such that $|\mathcal C| = \sum_{e\in E} \pi_e$, which gives $$p=\frac w2+\frac{\sum_{e\in E} \pi_e}{w-1}$$

Then, \DW-\Realizable{w}  can be formulated as an integer linear program:
\begin{align*}
  \min_{\substack{ X\in \{0,1\}^{p\times n}\\ y \in \{0,1\}^{|E|\times |\mathcal C|}}} 
    &\sum_{e \in E}  \sum_{i,j\in \mathcal C} y_{i,j}^{e} 
\\
\text{such that }\quad    \forall j \in [p]  \quad &\sum_{v=1}^{n} X_{v,j} = 1
\\
        \left. \begin{aligned}
      &\forall e=(v_1,v_2)  \in E  \\
    %&\forall e{'}=(v{'}_1, v{'}_2)  \notin E   \\ 
            &\forall (i,j)  \in \mathcal C
  \end{aligned}\right. \quad&
  \left\{\begin{aligned}
    -  X_{v_1, i}  &   && + y_{i,j}^{e}   \leq  0 \\
    &   - X_{v_2, j}  && + y_{i,j}^{e}    \leq  0 \\
    X_{v_1, i} & + X_{v_2, j}  &&  - y_{i,j}^{e}   \leq  1 \\
  \end{aligned}\right.
\\ 
    \left. \begin{aligned}
            &\forall e{'}=(v'_1, v'_2)  \notin E    \\ 
            &\forall (i,j)  \in \mathcal C  
  \end{aligned}\right. \quad&
  X_{v'_1 , i}  + X_{v'_2 , j}   \leq 1 
\\
    \mathrm{and} \; \; \forall e \in E 
         \quad&\sum_{(i,j)\in\mathcal C} y^{e}_{i,j}   \geq  \pi_e 
\end{align*}

If the objective function reaches $\sum_{e \in E} \pi_e$ at its minimum then the output of \DW-\Realizable{w}$(G, \Pi)$ is True, and False otherwise.

For the undirected variant, for each edge $e=\{u,v\}$ with $u\neq v$ we build the ILP as if we had two directed edges $(u,v)$ and $(v,u)$. The only specificity is for the edge weight  verification, that becomes the following:
$$ \forall e=\{u,v\} \in E , \sum_{(i,j)\in\mathcal C} y^{(u,v)}_{i,j} +y^{(v,u)}_{i,j}   \geq  \pi_e. $$
This ensures that arcs $(u,v)$ and $(v,u)$ get a total weight of $\pi_e$, that can be shared in any way. We thus get an ILP for   \GW-\Realizable{w}.

\subsubsection{Dynamic programming algorithm for \DW- and \GW-\NumRealizations{w}}\label{dynamic_programming_subsection}
We present here a baseline dynamic programming algorithm which, despite having punishing complexity, allows to compute the number of realizations for modest instances of the weighted directed and undirected cases.

The recursion proceeds by extending a partial sequence, initially set to be empty, keeping track along the way of represented edges and of the vertices appearing in the last window. 
Namely, consider $N_w[\Pi, p, \mathbf{u}]$ to be the number of $w$-realizations of length $p$ for the graph $G=(V,E)$, respecting a weight matrix
$\Pi=(\pi_{e})_{e \in E}$, preceded by a sequence of nodes $\mathbf{u}:= (u_1,\ldots, u_{|\mathbf{u}|}) \in V^{\star}$. It can be shown that, for all $p \ge 1$, $\Pi\in \mathbb{N}^{|E|}$ and $\mathbf{u}\in V^{\le w}$, $N_w[\Pi, p, \mathbf{u}]$ obeys the following formula in the directed case:
\begin{align*}
  N_w\left[\Pi, p, \mathbf{u}\right] &=  \sum_{\substack{v \in V }} 
  \begin{cases} 
    N_w\left[\Pi'_{(\mathbf{u},v)}, p-1, (u_1, ..., u_{|u|},v)\right] & \text{if }|\mathbf{u}| < w-1\\
    N_w\left[\Pi'_{(\mathbf{u},v)}, p-1, (u_2, ..., u_{w-1}, v )\right] & \text{if }|\mathbf{u}| = w-1\\
  \end{cases}\\
\text{with }\Pi'_{(\mathbf{u},v)}&:=(\pi'_e)_{e\in E}\\
\pi'_e &:= \pi_{e} - |\{k \in [1,|\mathbf{u}|] \mid e=(u_k,v) \}|
\end{align*}
The base case of this recurrence corresponds to $p=0$, and is defined as
\begin{equation}
  \forall \; \Pi, \; N_w[\Pi, 0, \mathbf{u}] =  \begin{cases}
    1 &\text{if } \Pi=(0)_{(i,j) \in V^2}\\
    0 &\text{otherwise.}
  \end{cases}
\end{equation}
The total number of realizations is then found in $N_w[\Pi, p, \varepsilon]$, \emph{i.e.}  setting $\mathbf{u}$ to the empty prefix $\varepsilon$, allowing the sequence to start from any node.

A similar dynamic programming scheme holds in the undirected case, through a minor modification ($e=\{u_k,v\}$ in the definition of $\pi'_e$).

The overall recurrence for sequence length $p$ can be computed in time  $\mathcal{O}( |V|^w \times \prod_{e\in E} (\pi_{e}+1))$ using memorization. 
This complexity can be refined by noting that: 
\begin{align*}
  \pi_{e}+1&\leq 2^{\pi_{e}} \\
\text{and } \sum_{e\in E} \pi_{e} &\leq w \times p\\
\text{so }\prod_{e\in E} (\pi_{e}+1)&\leq \prod_{e\in E} 2^{\pi_e}\\
&\leq 2^{w\,p}.
\end{align*}
We thus obtain an overall complexity in $\mathcal{O}(n^w 2^{w\,p})$, improving on the trivial $\mathcal O(n^p)$ enumeration algorithm whenever $n>2^w$. Then, despite the apparently high complexity of our algorithm, it is still possible to compute $N_w[\Pi, p, u_{1:w}]$ for ``reasonable'' values of $p$ and $w$. 
Precisely, succinct experiments showed that the table could be computed in less than a minute for values up to $|V|=20$, $p=100$ and $w=3$. See Figure~\ref{fig:resultDP} for an instance and the resulting sequences obtained by our algorithm. 

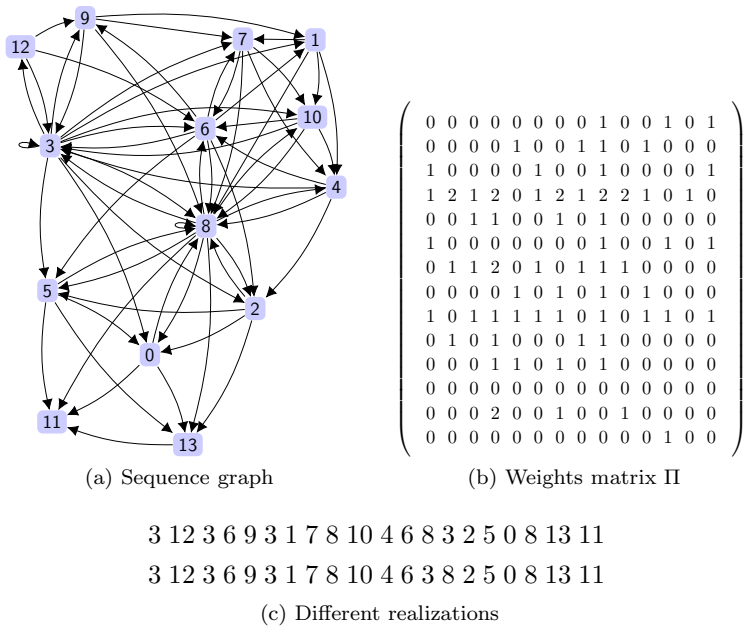
\begin{figure}
	\centering
	\subfloat[Sequence graph]{\scalebox{0.7}{
		\begin{tikzpicture}
        \tikzstyle{al}=[->, bend left=10]
        \tikzstyle{ar}=[->, bend right=10]
        \tikzstyle{ac}=[->]
		\begin{scope}[xshift=4cm]
		\node[main node] at (-4.793831,-4.990252) (0) {0};
		\node[main node] at (-1.650003,0.995279) (1) {1};
		\node[main node] at (-2.792707,-4.111722) (2) {2};
		\node[main node] at (-6.690637,-1.019216) (3) {3};
		\node[main node] at (-1.252901,-1.804257) (4) {4};
		\node[main node] at (-6.739747,-3.768199) (5) {5};
		\node[main node] at (-3.742480,-0.690621) (6) {6};
		\node[main node] at (-3.027863,1.011339) (7) {7};
		\node[main node] at (-3.723457,-2.544753) (8) {8};
		\node[main node] at (-6.023158,1.417405) (9) {9};
		\node[main node] at (-1.707117,-0.474145) (10) {10};
		\node[main node] at (-6.652897,-6.267174) (11) {11};
		\node[main node] at (-7.252901,0.863740) (12) {12};
		\node[main node] at (-4.072085,-6.698336) (13) {13};
		\draw[al] (0) edge node[right] {} (8);
		\draw[al] (0) edge node[right] {} (11);
		\draw[al] (0) edge node[right] {} (13);
		\draw[al] (1) edge node[right] {} (4);
		\draw[ac] (1) edge node[right] {} (7);
		\draw[ac] (1) edge node[right] {} (8);
		\draw[al] (1) edge node[right] {} (10);
		\draw[al] (2) edge node[right] {} (0);
		\draw[al] (2) edge node[right] {} (5);
		\draw[al] (2) edge node[right] {} (8);
		\draw[al] (2) edge node[right] {} (13);
		\draw[al] (3) edge node[right] {} (0);
		\draw[al] (3) edge node[right] {} (1);
		\draw[ar] (3) edge node[right] {} (2);
		\draw[al, loop above] (3.west) to (3.west);
		\draw[ar] (3) edge node[right] {} (5);
		\draw[al] (3) edge node[right] {} (6);
		\draw[al] (3) edge node[right] {} (7);
		\draw[al] (3) edge node[right] {} (8);
		\draw[al] (3) edge node[right] {} (9);
		\draw[al, bend left=15] (3) edge node[right] {} (10);
		\draw[al] (3) edge node[right] {} (12);
		\draw[al] (4) edge node[right] {} (2);
		\draw[al] (4) edge node[right] {} (3);
		\draw[al] (4) edge node[right] {} (6);
		\draw[al] (4) edge node[right] {} (8);
		\draw[al] (5) edge node[right] {} (0);
		\draw[al] (5) edge node[right] {} (8);
		\draw[ar] (5) edge node[right] {} (11);
		\draw[ar] (5) edge node[right] {} (13);
		\draw[ac] (6) edge node[right] {} (1);
		\draw[al] (6) edge node[right] {} (2);
		\draw[al] (6) edge node[right] {} (3);
		\draw[ar] (6) edge node[right] {} (5);
		\draw[al] (6) edge node[right] {} (7);
		\draw[al] (6) edge node[right] {} (8);
		\draw[ar] (6) edge node[right] {} (9);
		\draw[ar] (7) edge node[right] {} (4);
		\draw[al] (7) edge node[right] {} (6);
		\draw[al] (7) edge node[right] {} (8);
		\draw[al] (7) edge node[right] {} (10);
		\draw[al] (8) edge node[right] {} (0);
		\draw[al] (8) edge node[right] {} (2);
		\draw[al] (8) edge node[right] {} (3);
		\draw[al] (8) edge node[right] {} (4);
		\draw[al] (8) edge node[right] {} (5);
		\draw[al] (8) edge node[right] {} (6);
		\draw[al, loop above] (8.west) to (8.west);
		\draw[al] (8) edge node[right] {} (10);
		\draw[ar] (8) edge node[right] {} (11);
		\draw[al] (8) edge node[right] {} (13);
		\draw[al] (9) edge node[right] {} (1);
		\draw[al] (9) edge node[right] {} (3);
		\draw[ac] (9) edge node[right] {} (7);
		\draw[al] (9) edge node[right] {} (8);
		\draw[al, bend left=15] (10) edge node[right] {} (3);
		\draw[al] (10) edge node[right] {} (4);
		\draw[ac] (10) edge node[right] {} (6);
		\draw[al] (10) edge node[right] {} (8);
		\draw[al] (12) edge node[right] {} (3);
		\draw[al] (12) edge node[right] {} (6);
		\draw[al] (12) edge node[right] {} (9);
		\draw[al] (13) edge node[right] {} (11);
		\end{scope}
		\end{tikzpicture}

		}
	}
	\quad
	\subfloat[Weights matrix $\Pi$]{\scalebox{0.7}{

			\begin{tikzpicture}[>=stealth,thick]
			\matrix [matrix of math nodes,left delimiter=(,right delimiter=), ampersand replacement=\&](A){
				0 \& 0 \& 0 \& 0 \& 0 \& 0 \& 0 \& 0 \& 1 \& 0 \& 0 \& 1 \& 0 \& 1\\
				0 \& 0 \& 0 \& 0 \& 1 \& 0 \& 0 \& 1 \& 1 \& 0 \& 1 \& 0 \& 0 \& 0\\
				1 \& 0 \& 0 \& 0 \& 0 \& 1 \& 0 \& 0 \& 1 \& 0 \& 0 \& 0 \& 0 \& 1\\
				1 \& 2 \& 1 \& 2 \& 0 \& 1 \& 2 \& 1 \& 2 \& 2 \& 1 \& 0 \& 1 \& 0\\
				0 \& 0 \& 1 \& 1 \& 0 \& 0 \& 1 \& 0 \& 1 \& 0 \& 0 \& 0 \& 0 \& 0\\
				1 \& 0 \& 0 \& 0 \& 0 \& 0 \& 0 \& 0 \& 1 \& 0 \& 0 \& 1 \& 0 \& 1\\
				0 \& 1 \& 1 \& 2 \& 0 \& 1 \& 0 \& 1 \& 1 \& 1 \& 0 \& 0 \& 0 \& 0\\
				0 \& 0 \& 0 \& 0 \& 1 \& 0 \& 1 \& 0 \& 1 \& 0 \& 1 \& 0 \& 0 \& 0\\
				1 \& 0 \& 1 \& 1 \& 1 \& 1 \& 1 \& 0 \& 1 \& 0 \& 1 \& 1 \& 0 \& 1\\
				0 \& 1 \& 0 \& 1 \& 0 \& 0 \& 0 \& 1 \& 1 \& 0 \& 0 \& 0 \& 0 \& 0\\
				0 \& 0 \& 0 \& 1 \& 1 \& 0 \& 1 \& 0 \& 1 \& 0 \& 0 \& 0 \& 0 \& 0\\
				0 \& 0 \& 0 \& 0 \& 0 \& 0 \& 0 \& 0 \& 0 \& 0 \& 0 \& 0 \& 0 \& 0\\
				0 \& 0 \& 0 \& 2 \& 0 \& 0 \& 1 \& 0 \& 0 \& 1 \& 0 \& 0 \& 0 \& 0\\
				0 \& 0 \& 0 \& 0 \& 0 \& 0 \& 0 \& 0 \& 0 \& 0 \& 0 \& 1 \& 0 \& 0\\
			};
			\end{tikzpicture}
			
		}
	}
	
	\subfloat[Different realizations with $w=5$]{\makebox[5cm][c]{
			$\begin{aligned}
				3 \; 12 \; 3 \; 6 \; 9 \; 3 \; 1 \; 7 \; 8 \; 10 \; 4 \; 6 \; {\color{red}8} \; {\color{red}3 } \; 2 \; 5 \; 0 \; 8 \; 13 \; 11\\
				3 \; 12 \; 3 \; 6 \; 9 \; 3 \; 1 \; 7 \; 8 \; 10 \; 4 \; 6 \; {\color{red}3} \; {\color{red}8} \; 2 \; 5 \; 0 \; 8 \; 13 \; 11\\
			\end{aligned}$
		}
	}
	%\caption{$p=20$, $w=5$}
      \caption{Example of realizations in the  \DW variant, as obtained using our dynamic programming algorithm: (a) a 5-sequence graph on $|V|=14$ (vertices are labelled with integers from 0 to 13). (b) the corresponding weight matrix $\Pi$ of size $14\times14$. (c)  two possible realizations of length $p=20$.}
    \label{fig:resultDP}
\end{figure}

\section{Exponential lower bound on the size of realizations}\label{sec:exponential_realizations}
We have established in Section~\ref{sec:nphardnessproofs} the NP-hardness of several versions of realizability, yet have left open the question of their membership to NP. Such a property is usually proven by exhibiting a non-deterministic Turing machine which guesses a solution in polynomial time, and tests each of them in polynomial-time. This strategy requires that the (minimal) size of a solution grows only polynomially with the input length. Unfortunately, we find that the minimal size of a realization can grow exponentially larger that the input size, as formally stated below.

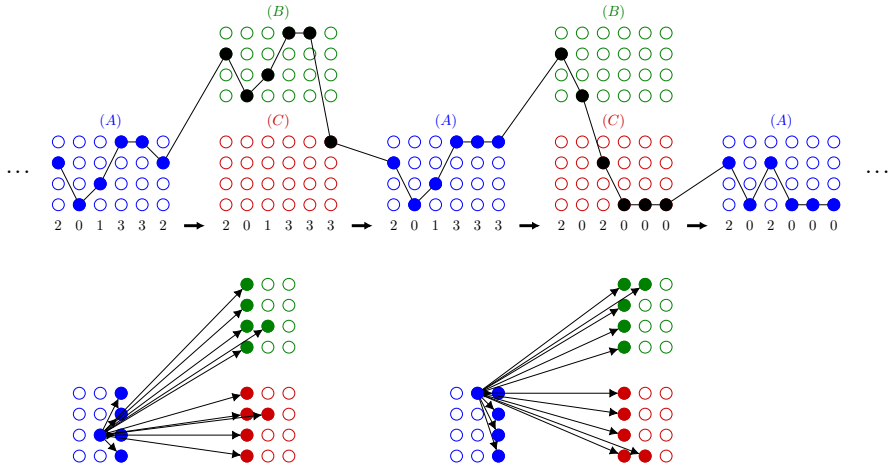
\begin{figure}
    \centering

%\documentclass{standalone}
%\usepackage{tikz}
%\usepackage{relsize}
%\usetikzlibrary{positioning}
%\usetikzlibrary{external}
%\usetikzlibrary{calc}
%\usetikzlibrary{fit}
%\usetikzlibrary{decorations.markings}

%\begin{document}
\begin{tikzpicture}[n/.style={draw, circle, inner sep=1mm}, A/.style={blue}, B/.style={green!50!black}, C/.style={red!80!black}, f/.style={n, fill}, X/.style={black}, x=5mm, y=5mm]

\def\k{4}
\def\n{6}
\def\deltaX{8}
\def\deltaY{5.2}
\def\dstep{16}

   \begin{scope}[shift={(0,-12)}]
\foreach\x/\y [count=\step] in {3/2,5/4} {
            \node[f, A] (ra\step) at (\x+\step*\dstep,\y) {};         
}
\foreach\x in {2,...,4} {
    \foreach\y in {1,...,\k} {
        \foreach \step/\dx in {1/0,2/2} {
            \node[n, A] (ra\step\x\y) at (\x+\step*\dstep+\dx,\y) {};         
            \node[n, B] (rb\step\x\y) at (\x+\deltaX+\step*\dstep+\dx,\y+\deltaY) {}; 
            \node[n, C] (rc\step\x\y) at (\x+\deltaX+\step*\dstep+\dx,\y) {}; 
        }
    }
}
\foreach \bs/\cs [count=\step] in {{21,22,23,24,32}/{21,22,23,24,33}, {21,22,23,24,34}/{21,22,23,24,31}}{
    \foreach \aval in {41,42,43,44}{
        \node[f,A, opacity=.4] at (ra\step\aval) {};
        \draw[->] (ra\step) -- (ra\step\aval);
    }
    \foreach \b in \bs {
    \node[f,B, opacity=.4] at (rb\step\b) {};
    \draw[->] (ra\step) -- (rb\step\b);
    }
    \foreach \c in \cs {
    \node[f,C, opacity=.4] at (rc\step\c) {};
    \draw[->] (ra\step) -- (rc\step\c);
    }
    %\node at (rb\step\bs) {X};
}

   \end{scope} 
\node at (56,2.6) { \large \textbf{ \dots}};
\node at (15,2.6) { \large \textbf{ \dots}};
\foreach\x in {1,...,\n} {
    \foreach\y in {1,...,\k} {
        \foreach \step in {1,2,3} {
            \node[n, A] (a\step\x\y) at (\x+\step*\dstep,\y) {};         
        }
        \foreach \step in {1,2} {
            \node[n, B] (b\step\x\y) at (\x+\deltaX+\step*\dstep,\y+\deltaY) {}; 
            \node[n, C] (c\step\x\y) at (\x+\deltaX+\step*\dstep,\y) {}; 
        }
    }
}

\foreach \step in {1,2,3} {
    \node[A] at (\n/2+1/2+\dstep*\step,\k+1) {($A$)};
}

\foreach \step in {1,2} {
    \node[B] at (\n/2+1/2+\dstep*\step + \deltaX,\k+1+\deltaY) {($B$)};
    \node[C] at (\n/2+1/2+\dstep*\step + \deltaX,\k+1) {($C$)};
}
\foreach \u [count=\x] in {3,1,2,4,4,3} {
    \node[f,A] at (a1\x\u) (p1\x) {};    
    \pgfmathtruncatemacro{\v}{\u-1}
    \node at ($(a1\x1)+(0,-1)$) {\v};
}

\foreach \u [count=\x] in {3,1,2,4,4} {
    \node[f,X] at (b1\x\u) (p2\x) {};       
    \pgfmathtruncatemacro{\v}{\u-1}
    \node at ($(c1\x1)+(0,-1)$) {\v};
}
\foreach \u [count=\x from 6] in {4} {
    \node[f,X] at (c1\x\u) (p2\x) {};       
    \pgfmathtruncatemacro{\v}{\u-1}
    \node at ($(c1\x1)+(0,-1)$) {\v};    
}

\foreach \u [count=\x] in {3,1,2,4,4,4} {
    \node[f,A] at (a2\x\u) (p3\x) {};    
    \pgfmathtruncatemacro{\v}{\u-1}
    \node at ($(a2\x1)+(0,-1)$) {\v};    
}

\foreach \u [count=\x] in {3,1} {
    \node[f,X] at (b2\x\u) (p4\x) {};    
    \pgfmathtruncatemacro{\v}{\u-1}
        
    \node at ($(c2\x1)+(0,-1)$) {\v};
}
\foreach \u [count=\x from 3] in {3,1,1,1} {
    \node[f,X] at (c2\x\u) (p4\x) {};    
        
    \pgfmathtruncatemacro{\v}{\u-1}
    \node at ($(c2\x1)+(0,-1)$) {\v};
}
\foreach \u [count=\x] in {3,1,3,1,1,1} {
    \node[f,A] at (a3\x\u) (p5\x) {};    
        
    \pgfmathtruncatemacro{\v}{\u-1}
    \node at ($(a3\x1)+(0,-1)$) {\v};
}

\foreach \p/\c in {1/A,2/X,3/A,4/X,5/A} {
    \draw [] (p\p1)--(p\p2)--(p\p3)--(p\p4)--(p\p5)--(p\p6);
}
\begin{scope} []
\draw (p16) -- (p21);
\draw (p26) -- (p31);
\draw (p36) -- (p41);
\draw (p46) -- (p51);
\end{scope}
\begin{scope} [line width=2pt]
\foreach \step/\nxstep in {1/2,2/3} {
    \draw[->] ($(a\step61)+(1,-1)$) -- ($(c\step11)+(-1,-1)$) ;
    \draw[->] ($(c\step61)+(1,-1)$) -- ($(a\nxstep11)+(-1,-1)$) ;
}

\end{scope}

\end{tikzpicture}

%\end{document}
    \caption{Illustration of the construction in Proposition \ref{prop:exponential_realizations} for a graph with an exponentially long realization, with $n=4$ and $k=6$. Top: a fragment of the path, starting with the substring $a_{1,2}a_{2,0}  a_{3,1}a_{4,3}a_{5,3} a_{6,2}$: in a correct realization,  such a fragment (with value $(2,0,1,3,3,2)$) must be followed by the highlighted vertices with successive values $(2,0,1,3,3,3)$ and $(2,0,2,0,0,0)$. Vertices are drawn multiple times, to avoid overlappings in the drawing, but there are indeed only $n\times k$ vertices in each of $A$, $B$ an $C$. This counting behavior must be repeated from $(0,0,0,0,0,0)$ to $(3,3,3,3,3,3)$, yielding a path of length at least $4^6$. Bottom: example of arcs outgoing from two $A$ vertices that enforce this behavior. Each vertex in $A$ is connected to a single vertex in the corresponding column in each of $B$ and $C$, where $B$ is used to keep the same value and $C$ is used to increment a column.}
    \label{fig:expo}
\end{figure}

\begin{proposition}\label{prop:exponential_realizations} 
    For any positive integers $n$ and $k$, there exists a graph of size $3kn+1$ such that any \DU-realization with a window of size $k+1$ has length at least $2kn^k$.
\end{proposition}

\begin{proof} See Figure~\ref{fig:expo} for an example.
Our construction uses three sets of vertices $A$, $B$ and $C$ of size  $k\times n$ each (vertices are labeled respectively $a_{i,j}$, $b_{i,j}$ and $c_{i,j}$ with $ i\in[k]$ and $0\leq j<n$), plus an additional \emph{start} vertex $s$. 
A vertex   $a_{i,j} \in A$ has \emph{rank} $i$ and \emph{value} $j$.
A vertex  $b_{i,j} \in B$  or $c_{i,j}\in C$  has \emph{rank} $i+k$ and \emph{value} $j$. 
Vertex $s$ has rank and value 0. Ranks are counted in $\mathbb Z/2k\mathbb Z$. 

We consider $k$-tuples $T=(j_1,\ldots,j_k)$ with values in $[0,n-1]$. They are ordered according to the lexicographic order from $(0,\ldots,0)$, $(0,\ldots,1)$ to $(n-1,\ldots n-1)$. In particular, the successor of $T$ is the $k$-tuple $T'=(j_1, \ldots, j_{x-1}, j_x+1,0, \ldots, 0)$ where $x$ is the largest index such that $j_x<n-1$. 

We build a DAG on vertex set $A\cup B\cup C \cup \{s\}$ with the following arcs. 
Vertex $s$ has outgoing arcs to each of $a_{i,0}$ for all $i$.
Each vertex $a_{i,j}$ with $1\leq i \leq k$ and $0\leq j <n$ has an outgoing arc 
to each $a_{i',j'}$ with $i'>i$,
to each $b_{i',j'}$ with $i'<i$,
to $b_{i,j}$ and to $c_{i,j+1\mod n}$.
Each vertex  $b_{i,j}$ with $1\leq i \leq k$ and $0\leq j <n$ has an outgoing arc 
 to each $b_{i',j'}$ with $i'>i$, 
 to each $a_{i',j'}$ with $i'<i$ 
 and to $a_{i,j}$.
Finally, each $c_{i,j}$ with $1\leq i \leq k$ and $0\leq j <n$ has an outgoing arc 
 to $c_{i',0}$ for $i'>i$,
 to each $a_{i',j'}$ with $i'<i$ 
 and to $a_{i,j}$.

Let $S$ be a realization of $G$ with window size $k+1$. Clearly $S$ necessarily starts with $s$ (the only vertex with in-degree 0). Let $1\leq p\leq |S|-k$. Consider the substring $S'=S[p \ldots p+k]$. 
Note that by construction a vertex of rank $r$ only has outgoing arcs to vertices with rank $r+i$ with $0<i\leq k$. In particular, two vertices of the same rank cannot be in $S'$. Thus, let $r$ be the rank of $S[p]$, then all other vertices  of $S'$ have rank in $[r+1, r+k]$. In particular, the second vertex $S[p+1]$ in $S'$ has out-going arcs to $k-1$ vertices with distinct ranks among $[r+1, r+k]$, which is only possible for vertices of rank $r-1$, $r$, or $r+1$. Thus $S[p+1]$ has necessarily rank $r+1$. Hence, since $S[1]=s$ has rank $0$, then $S[i]$ has rank $i-1$ for $1\leq i \leq |S|-k$. In particular, $S[1,\ldots, k+1]=s a_{0,0}\ldots a_{k,0}$.

Let $a_{i,j}\in A$ and $p$ such that $S[p]=a_{i,j}$. Then $S[p+k]$ is one of $b_{i,j}, c_{i,j+1}$. For $S[p]=b_{i,j}\in B$ or $S[p]=c_{i,j}\in C$ then $S[p+k]= a_{i,j}$. Thus, in most cases, the value of $S[p]$ and $S[p+k]$ are equal, except in the case where $S[p]=a_{i,j}$ and $S[p+k] = c_{i,j+1}$. 
Then by the outgoing arcs of $c_{i,j+1}$, necessarily $S[p+k+i'-i]=c_{i', 0}$ for all $i<i'\leq k$.
%In particular, $S[p]=c_{i,0}$ and $S[p+k] = a_{i,0}$. 
Let $p$ be a position such that $S[p]$ has rank 1, let $T=(j_1,\ldots, j_k)$ be the tuple of values of $S[p]\ldots S[p+k-1]$, let $T'$ be the tuple of values of $S[p+k]\ldots S[p+2k-1]$, and $T''$ be the tuple of values of $S[p+2k]\ldots S[p+2k-1]$. Then if $S[p+k]\ldots S[p+2k-1]$ does not contain any vertex in $C$, then $T=T'=T''$. Otherwise,  $T'=(j_1, \ldots, j_{x-1}, (j_x+1 \mod n),0, \ldots, 0)$ with $x$ the smallest index such that $S[p+k+x]\in C$. In particular,  $S[p+k+i']=c_{i',0}$ and $S[p+i']=a_{i',n-1}$ for each $i'$ with $x<i'<k$. 
If $j_x=n-1$, then $T'$ is before $T$ in lexicographic order (same prefix, and all values in an entire suffix is reset to 0). Otherwise $j_x<n-1$ so $x$ is the largest index such that $j_x<n-1$ and $T'$ is the successor of $T$

To conclude, $S$ contains $a_{0,0}\ldots a_{k,0}$,  i.e. a substring with tuple of values $(0,\ldots,0)$. It also uses the arc $(a_{1,n-1}, c_{1,0})$, which implies that $S$ also contains $a_{1,n-1}\ldots a_{k,n-1}c_{1,0} \ldots c_{k,0}$, hence $S$ contains a substring with value tuple $(n-1\ldots n-1)$. For any two successive value tuples $T, T'$ either $T'$ is lower (or equal) to $T$ in lexicographical order, either $T'$ is the successor of $T$, so overall $S$ must use every possible tuple at least once (for substrings with ranks $1$ to $k$). Thus $S$ has length at least $(2k)n^k$. 

Note that the above proof does not guarantee the actual \emph{existence} of such a realization. However, the construction can be adapted to this end, by providing an exponential-length sequence using only arcs from the DAG (starting with $s a_{1,0}\ldots a_{k,0}$ and ending with $a_{1,n-1}\ldots a_{k,n-1}$), and filtering out those edges that are not realized. Thus, any sequence realizing the resulting graph still requires an exponential length, and the graph is realizable by construction.
\end{proof}
\section{Discussion }
In this study, we have formalized a new series of inverse problems to assert the (un)ambiguity of popular word embeddings. We have provided a comprehensive characterization of their complexity, leaving only open their general membership in NP whenever $w\geq 3$. 
{Indeed, given a sequence, computing its sequence graph representation can be done in $O(d^2+p)$, if $p$ is the length of the sequence and $d$ the size of the vocabulary}. However, this does not prove that \Realizable{} nor \NumRealizations{} are in NP, because the said realization could be exponentially large with respect to the number of vertices or the window size. Although we cannot settle this question in general, we proved that this situation occurs in the directed case (\DU and \DW), for which some graphs have minimal realizations whose length scales exponentially with the window size. This is formally stated in Proposition \ref{prop:exponential_realizations} for \DU-\Realizable{}.

%\todo[inline]{conclude with LLMs?}
Given the success of Large Language Models (LLMs), it would be of interest to consider  similar inverse problems for LLM embeddings. In particular, this may lead to a better understand the ability of the LLMs to encode syntactic and semantic information. We think that the ones studied in this article would probably need to be adapted in order to gain some meaningful insights. Due to their over-parametrization, we suspect that the map between sequences and embeddings in that case is injective (for sequences up to a few thousand of symbols).  In particular, \NumRealizations{} would become equivalent to \Realizable{}. %More generally, a challenging question of interest are 
 More broadly, a challenging question of interest suggested by this study are the potential  connections  between the computational complexity of  (well-chosen) inverse problems of embeddings (such as \NumRealizations{}, \Realizable{}, or variants thereof) and the capacity of the considered embedding to faithfully capture syntactic and semantic properties of natural language. 

 \bigskip

%\section*{Acknowledgments}

\noindent \textbf{Acknowledgments:} The authors wish to express their gratitude to  Guillaume Fertin and anonymous reviewers of an earlier version of this manuscript, for their valuable suggestions and constructive criticisms. 
\medskip

\noindent \textbf{Author Contributions:} All authors contributed equally to this work.

\medskip

\noindent \textbf{Funding:} No funding to declare.

\end{document}